\documentclass[11pt]{elsarticle}

\makeatletter  
\def\ps@pprintTitle{%
 \let\@oddhead\@empty
 \let\@evenhead\@empty
 \def\@oddfoot{\centerline{\thepage}}%
 \let\@evenfoot\@oddfoot}
\makeatother

\makeatletter
\def\@author#1{\g@addto@macro\elsauthors{\normalsize%
    \def\baselinestretch{1}%
    \upshape\authorsep#1\unskip\textsuperscript{%
      \ifx\@fnmark\@empty\else\unskip\sep\@fnmark\let\sep=,\fi
      \ifx\@corref\@empty\else\unskip\sep\@corref\let\sep=,\fi
      }%
    \def\authorsep{\unskip,\space}%
    \global\let\@fnmark\@empty
    \global\let\@corref\@empty  
    \global\let\sep\@empty}%
    \@eadauthor={#1}
}
\makeatother

\usepackage{fullpage}
\usepackage[english]{babel}
\usepackage{enumerate}
\usepackage[colorlinks=true,linkcolor=cyan]{hyperref}
\usepackage{amsmath,amsthm,amsfonts,amssymb}

\usepackage{float}
\usepackage{setspace}

\newcommand\xqed[1]{
  \leavevmode\unskip\penalty9999 \hbox{}\nobreak\hfill
  \quad\hbox{#1}}
\newcommand{\qedexmp}{\xqed{$\diamond$}}

\floatstyle{ruled}
\newfloat{algorithm}{tbp}{loa}
\floatname{algorithm}{Algorithm}

\usepackage{algorithmic}

\newcommand{\Z}{\mathbb{Z}}
\newcommand{\N}{\mathbb{N}}

\newcommand{\cdeg}[2][]{{\rm cdeg}_{#1}(#2)}
\newcommand{\rdeg}[2][]{{\rm rdeg}_{#1}(#2)}
\newcommand{\diag}{{\rm Diag}}

\DeclareMathOperator{\hermiteDiagonal}{HermiteDiagonal}
\DeclareMathOperator{\colBasis}{ColumnBasis}
\DeclareMathOperator{\determinant}{determinant}
\DeclareMathOperator{\mkb}{MinimalKernelBasis}

\newcommand{\mat}[1]{\mathbf{#1}}  
\newcommand{\row}[1]{\mathbf{#1}} 
\newcommand{\col}[1]{\mathbf{#1}} 
\newcommand{\trsp}[1]{#1^\mathsf{T}} 
\newcommand{\var}{x}   

\newcommand{\matdim}{n}  
\newcommand{\hermite}{\mat{H}}  
\newcommand{\smith}{\mat{S}}  
\newcommand{\idMat}{\mathbf{I}}  
\newcommand{\degDet}[1][A]{D(\mat{#1})}  
\newcommand{\bigO}[1]{\mathcal{O}(#1)} 
\newcommand{\softO}[1]{\widetilde{\mathcal{O}}(#1)} 
\newcommand{\expmatmul}{\omega}  

\newcommand{\field}{\mathbb{K}} 
\newcommand{\polRing}{\field[\var]} 

\newcommand{\storeArg}{} 

\newcommand{\matSpace}[1][\matdim]{\renewcommand\storeArg{#1}\matSpaceAux} 
\newcommand{\polMatSpace}[1][\matdim]{\renewcommand\storeArg{#1}\polMatSpaceAux} 
\newcommand{\matSpaceAux}[1][\storeArg]{\field^{\storeArg \times #1}} 
\newcommand{\polMatSpaceAux}[1][\storeArg]{\polRing^{\storeArg \times #1}} 

\newcommand{\shift}[2][s]{#1_{#2}} 
\newcommand{\shifts}[1][s]{\tuple{#1}} 
\newcommand{\shiftSpace}[1][\matdim]{\Z^{#1}} 
\newcommand{\unishift}{\shifts[0]} 
\newcommand{\leadingMat}[2][]{\mathrm{lm}_{#1}(#2)} 
\newcommand{\shiftMat}[1]{\mat{\var}^{#1\,}} 

\newcommand*{\vv}[1]{\vec{\mkern0mu#1}}  

\newcommand{\anyMat}{\boldsymbol{\ast}}
\newcommand{\rowParLinDiag}[2][\minDegs]{\mathcal{L}_{#1}(#2)}
\newcommand{\degExp}{\delta}
\newcommand{\shiftDegExp}{\shifts[m]}

\newcommand{\expandMat}[1][\degLins]{\boldsymbol{\mathcal{E}}_{#1}}
\newcommand{\triExpMat}[1][\degLins]{\boldsymbol{\mathcal{T}}_{#1}}
\newcommand{\quoExp}{\alpha}
\newcommand{\remExp}{\beta}
\newcommand{\expand}[1]{\widetilde{#1}}
\newcommand{\minDeg}{\delta}
\newcommand{\minDegs}{\vv{\delta}}
\newcommand{\tuple}[1]{\vv{#1}}   
\newcommand{\sumVec}[1]{|#1|}

\usepackage[capitalize]{cleveref}
\crefname{equation}{equation}{equations}
\Crefname{equation}{Equation}{Equations}

\theoremstyle{plain}
\newtheorem{definition}{Definition}[section]
\newtheorem{theorem}[definition]{Theorem}
\newtheorem{corollary}[definition]{Corollary}
\newtheorem{proposition}[definition]{Proposition}
\newtheorem{lemma}[definition]{Lemma}
\theoremstyle{remark}
\newtheorem{example}[definition]{Example}
\newtheorem{remark}[definition]{Remark}

\begin{document}

\begin{frontmatter}
\title{Fast, deterministic computation of the Hermite normal form \\ and determinant of a polynomial matrix}

\author{George Labahn\corref{cor}\fnref{aff1}}
\ead{glabahn@uwaterloo.ca}

\fntext[aff1]{David R. Cheriton School of Computer Science, University of Waterloo, Waterloo ON, Canada N2L 3G1}
\cortext[cor]{Corresponding author.}

\author{Vincent Neiger\fnref{aff2}}
\ead{vincent.neiger@ens-lyon.fr}

\author{Wei Zhou\fnref{aff1}}
\ead{w2zhou@uwaterloo.ca}

\fntext[aff2]{ENS de Lyon (Laboratoire LIP, CNRS, Inria, UCBL,
Universit\'e de Lyon), Lyon, France}


\begin{abstract}
Given a nonsingular $\matdim \times \matdim$ matrix of univariate polynomials
over a field $\field$, we give fast and deterministic algorithms to compute its determinant and its
Hermite normal form. Our algorithms use
$\softO{\matdim^\expmatmul \lceil s \rceil}$ operations in $\field$, where $s$
is bounded from above by both the average of the degrees of the rows and that
of the columns of the matrix and $\expmatmul$ is the exponent of matrix
multiplication. The soft-$\mathcal{O}$ notation indicates that logarithmic factors in the
big-$\mathcal{O}$ are omitted while the ceiling function indicates that the cost is
$\softO{\matdim^\expmatmul}$ when $s = o(1)$.  Our algorithms are based on a
fast and deterministic triangularization method for computing the diagonal
entries of the Hermite form of a nonsingular matrix.
\end{abstract}


\end{frontmatter}

\section{Introduction}

For a given nonsingular polynomial matrix $\mat{A}$ in $\polMatSpace$, one can
find a unimodular matrix $\mathbf{U} \in \polMatSpace$ such that $\mat{A}
\mat{U} = \mat{H}$ is triangular. Unimodular means that there is a polynomial
inverse matrix, or equivalently, the determinant is a nonzero constant from
$\field$. Triangularizing a matrix is useful for solving linear systems and
computing matrix operations such as determinants or normal forms.  In the
latter case, the best-known example is the Hermite normal form, first defined
by Hermite in 1851 in the context of triangularizing integer
matrices~\cite{Hermite1851}. Here,
\[ \hermite =
\left[ \begin{array}{cccc} 
    h_{11} & & \\
    h_{21} & h_{22} & & \\
    \vdots & \vdots & \ddots & \\
  h_{\matdim 1} & \cdots & \cdots & h_{\matdim \matdim}
\end{array} \right] \]
with the added properties that each $h_{ii}$ is monic and $\deg(h_{ij}) <
\deg(h_{ii})$ for all $j<i$. Classical variations of this definition include
specifying upper rather than lower triangular forms, and specifying row rather
than column forms. In the latter case, the unimodular matrix multiplies on the
left rather than the right, and the degree of the diagonal entries dominates
that of their columns rather than their rows.

The goal of this paper is the fast, deterministic computation of the
determinant and Hermite normal form of a nonsingular polynomial matrix. The
common ingredient in both algorithms is a method for the fast computation of the
diagonal entries of a matrix triangularization. The product of these entries
gives, at least up to a constant, the determinant while Hermite forms are
determined from a given triangularization by reducing the remaining entries
modulo the diagonal entries.

In the case of determinant computation, there has been a number of efforts
directed to obtaining algorithms whose complexities are given in terms of
exponents of matrix multiplication. Interestingly enough, in the case of
matrices over a field,  Bunch and Hopcroft \cite{BunchHopcroft1974} showed that
if there exists an algorithm which multiplies $\matdim\times\matdim$ matrices
in $\bigO{\matdim^\expmatmul}$ field operations for some $\expmatmul$, then
there also exists an algorithm for computing the determinant with the same cost
bound $\bigO{\matdim^\expmatmul}$. In the case of an arbitrary commutative ring
or of the integers, fast determinant algorithms have been given by Kaltofen
\cite{kaltofen92}, Abbott \emph{et al.} \cite{abbott}
and Kaltofen and Villard
\cite{KaltofenVillard}. We refer the reader to the last named paper and the
references therein for more details on efficient determinant computation of
such matrices.

In the specific case of the determinant of a matrix of polynomials $\mat{A}$
with $\deg(\mat{A})=d$, Storjohann \cite{storjohann:phd2000} gave a recursive
deterministic algorithm making use of fraction-free Gaussian elimination with a
cost of $\softO{\matdim^{\expmatmul + 1} d}$ operations. A deterministic
$\bigO{\matdim^3 d^2}$ algorithm was later given by Mulders and Storjohann
\cite{mulders-storjohann:2003}, modifying their algorithm for weak Popov form
computation.  Using low rank perturbations, Eberly \emph{et al.}
\cite{EberlyGiesbrechtVillard} gave a randomized determinant algorithm for integer matrices
which can be adapted to be used with polynomial matrices using
$\softO{\matdim^{3.5} d}$ field operations. Storjohann
\cite{storjohann:2003} later used high order lifting to give a randomized
algorithm which computes the determinant using $\softO{\matdim^\expmatmul d}$
field operations. The algorithm of Giorgi \emph{et al.} \cite{Giorgi2003} has a
similar cost but only works on a class of generic input matrices, matrices that
are well behaved in the computation.

Similarly there has been considerable progress in the efficient computation of
the Hermite form of a polynomial matrix. Hafner and McCurley \cite{hafner} and
Iliopoulos \cite{iliopoulos} give algorithms with a complexity bound of
$\softO{\matdim^{4} d}$ operations from $\field$ where $d = \deg(\mat{A})$.
They control the size of the matrices encountered during the computation by
working modulo the determinant. Using matrix multiplication the algorithms of
Hafner and McCurley \cite{hafner}, Storjohann and Labahn
\cite{storjohann-labahn96} and Villard \cite{villard96} reduce the cost to
$\softO{\matdim^{\expmatmul + 1}d}$ operations where  $\expmatmul$ is the
exponent of matrix multiplication. The algorithm of Storjohann and Labahn
worked with integer matrices but the results directly carry over to polynomial
matrices. Mulders and Storjohann \cite{mulders-storjohann:2003} then gave an
iterative algorithm having complexity $\bigO{\matdim^{3} d^2}$, thus reducing
the exponent of $\matdim$ but at the cost of increasing the exponent of~$d$. 

During the past two decades, there has been a goal to design algorithms that
perform various $\polRing$-linear algebra operations in about the time that it
takes to multiply two polynomial matrices having the same dimension and degree
as the input matrix, namely at a cost $\softO{\matdim^{\expmatmul}d}$.
\emph{Randomized} algorithms with such a cost already exist for a number of
polynomial matrix problems, for example for linear system solving
\cite{storjohann:2003}, Smith normal form computation \cite{storjohann:2003},
row reduction \cite{Giorgi2003} and small nullspace bases computation
\cite{StoVil05}. In the case of polynomial matrix inversion, the randomized
algorithm in \cite{Storjohann15} costs $\softO{\matdim^3 d}$, which is
quasi-linear in the number of field elements used to represent the inverse.
For Hermite form computation, Gupta and Storjohann \cite{GS2011} gave a
randomized algorithm with expected cost $\softO{\matdim^{3}d}$, later improved
to $\softO{\matdim^\expmatmul d}$ in~\cite{Gupta11}. Their algorithm was the
first to be both softly cubic in $\matdim$ and softly linear in $d$. It is
worth mentioning that all the algorithms cited in this paragraph are of the Las
Vegas type.

Recently, \emph{deterministic} fast algorithms have been given for linear
system solving and row reduction \cite{GSSV2012}, minimal nullspace bases
\cite{za2012}, and matrix inversion \cite{ZhLaSt15}. Having a deterministic
algorithm has advantages. As a simple but important example, this allows for
use over a small finite field $\field$ without the need for resorting to
field extensions. The previous fastest Hermite form
algorithms~\cite{GS2011,Gupta11} do require such field extensions. In this
paper, we give deterministic fast algorithms for computing Hermite forms and
determinants.

Our approach relies on an efficient method for determining the diagonal
elements of a triangularization of the input matrix $\mat{A}$. We can do this
recursively by determining, for each integer $k$, a partition
\begin{equation}
  \label{eqn:diag_recursion}
\mat{A} \cdot \mat{U}=\begin{bmatrix}\mat{A}_{u}\\
\mat{A}_{d}
\end{bmatrix}\begin{bmatrix}\mat{U}_{\ell} & \mat{U}_{r}\end{bmatrix}=\begin{bmatrix}\mat{B}_{1} & \mat{0} \\
\anyMat & \mat{B}_{2}
\end{bmatrix}=\mat{B} \nonumber
\end{equation}
where $\mat{A}_{u}$ has $k$ rows, $\mat{U}_{\ell}$ has $k$ columns and
$\mat{B}_{1}$ is of size $k \times k$. The subscripts for $\mat{A}$ and
$\mat{U}$ are meant to denote up, down, left and right. As $\mat{A}$ is
nonsingular, $\mat{A}_{u}$ has full rank and hence one has that $\mat{U}_{r}$
is a basis of the kernel of $\mat{A}_{u}$.  Furthermore the matrix
$\mat{B}_{1}$ is nonsingular and is therefore a column basis of $\mat{A}_{u}$. 

However the recursion described above requires additional properties if it is to be efficient
for our applications. In the case of determinants, $\mat{A} \cdot \mat{U}$
being lower triangular implies that we need both the product of the diagonals
and also the determinant of the unimodular multiplier. For the case of Hermite
form computation a sensible approach would be to first determine a triangular
form of $\mat{A}$ and then reduce the lower triangular elements using the
diagonal entries with unimodular operations.  In both applications it appears
that we would need to know $\mat{U} = \mat{A}^{-1} \mat{H}$.  However the
degrees in such a unimodular multiplier can be too large for efficient
computation. Indeed there are examples where the sum of the degrees in
$\mat{U}$ is $\Theta(\matdim^3 d)$ (see \Cref{sec:diagonals}), in which case
computing $\mat{U}$ is beyond our target cost $\softO{\matdim^\expmatmul d}$. 

In order to achieve the desired efficiency, our triangularization computations
need to be done without actually determining the entire unimodular matrix
$\mat{U}$. We accomplish this by making use of shifted minimal kernel bases and
column bases of polynomial matrices, whose computations can be done efficiently
using algorithms from \cite{za2012} and \cite{za2013}. Shifts are weightings of
column degrees which basically help us to control the computations using column
degrees rather than the degree of the polynomial matrix. Using the degree
becomes an issue for efficiency when the degrees in the input matrix vary
considerably from column to column. We remark that shifted minimal kernel bases
and column bases, used in the context of fast block elimination, have also been
used for deterministic algorithms for inversion \cite{ZhLaSt15} and unimodular
completion \cite{zl2014} of polynomial matrices.

Fast algorithms for computing shifted minimal kernel bases~\cite{za2012} and
column bases~\cite{za2013} imply that we can deterministically find the
diagonals in $\softO{\matdim^{\expmatmul} \lceil s\rceil}$ field operations,
where $s$ is the average of the column degrees of $\mat{A}$. We recall that the
ceiling function indicates that for matrices with very low average column
degree $s \in o(1)$, this cost is still $\softO{\matdim^\expmatmul}$. By
modifying this algorithm slightly we can also compute the determinant of the
unimodular multiplier, giving our first contribution. In the next theorem, $\degDet$ is the
so-called \emph{generic determinant bound} as defined in \cite{GSSV2012} (see
also \Cref{subsec:degdet}). It has the important property that
$\degDet/\matdim$ is bounded from above by both the average of the degrees of
the columns of $\mat{A}$ and that of its rows. 

\begin{theorem}
  \label{thm:determinant}
  Let $\mat{A}$ be a nonsingular matrix in $\polMatSpace$. There is a
  deterministic algorithm which computes the determinant of $\mat{A}$ using
  $\softO{\matdim^\expmatmul \lceil \degDet/\matdim \rceil} \subseteq
  \softO{\matdim^\expmatmul \lceil s \rceil}$ operations in $\field$, with $s$
  being the minimum of the average of the degrees of the columns of $\mat{A}$
  and that of its rows.
\end{theorem}

Applying our fast diagonal entry algorithm for Hermite form computation has
more technical challenges. The difficulty comes
from the unpredictability of the diagonal degrees of $\hermite$, which coincide
with its row degrees. Indeed, we know that the sum of the diagonal degrees in
$\hermite$ is $\deg(\det(\mat{A})) \le \matdim d$, and so the sum of the
degrees in $\hermite$ is $\bigO{\matdim^2 d}$. Still, the best known \emph{a
priori} bound for the degree of the $i$-th diagonal entry is $(\matdim - i + 1)
d$ and hence the sum of these bounds is $\bigO{\matdim^2 d}$, a factor of
$\matdim$ larger than the actual sum. Determining the diagonal entries gives us
the row degrees of $\hermite$ and thus solves this issue. Still, it remains a
second major task: that of computing the remaining entries of $\hermite$.

The randomized algorithm of Gupta and Storjohann~\cite{GS2011,Gupta11}
solves the Hermite form problem using two steps, which both make use of the
Smith normal form $\smith$ of $\mat{A}$ and partial information on a left
multiplier $\mat{V}$ for this Smith form. The matrices $\smith$ and $\mat{V}$
can be computed with a Las Vegas randomized algorithm using an expected number
of $\softO{\matdim^\expmatmul d}$ field operations~\cite{GS2011,Gupta11},
relying in particular on high-order lifting~\cite[Section~17]{storjohann:2003}.
The first step of their algorithm consists of computing the diagonal entries of
$\hermite$ by triangularization of a $2\matdim \times 2\matdim$ matrix
involving $\smith$ and $\mat{V}$, a computation done in
$\softO{\matdim^\expmatmul d}$ operations~\cite{Gupta11}. The second step sets
up a system of linear modular equations which admits $\mat{A}$ as a basis of
solutions: the matrix of the system is $\mat{V}$ and the moduli are the
diagonal entries of $\smith$. The degrees of the diagonal entries obtained in
the first step are then used to find $\hermite$ as another basis of solutions
of this system, computed in $\softO{\matdim^\expmatmul d}$~\cite{GS2011} using
in particular fast minimal approximant basis and partial linearization
techniques~\cite{Storjohann06,ZL2012}.

The algorithm presented here for Hermite forms follows a two-step process
similar to the algorithm of Gupta and Storjohann, but it avoids using the Smith
form of $\mat{A}$, whose deterministic computation in
$\softO{\matdim^\expmatmul d}$ still remains an open problem. Instead, as
explained above, we compute the diagonal entries of $\hermite$
deterministically via \cref{eqn:diag_recursion} using
$\softO{\matdim^{\expmatmul} \lceil s\rceil}$ field operations, where $s$ is
the average of the column degrees of $\mat{A}$.  As for the second step, using
the knowledge of the diagonal degrees of $\hermite$ combined with partial
linearization techniques from~\cite[Section~6]{GSSV2012}, we show that
$\hermite$ can then be computed via a single call to fast deterministic column
reduction~\cite{GSSV2012} using $\softO{\matdim^{\expmatmul} d}$ field
operations. This new problem reduction illustrates the fact that
knowing in advance the degree shape of reduced or normal forms makes their
computation much easier, something already observed and exploited
in~\cite{GS2011,zhou:phd2012,JeNeScVi16}.

This approach results in a deterministic $\softO{\matdim^{\expmatmul} d}$
algorithm for Hermite form computation, which is satisfactory for matrices
$\mat{A}$ that have most entries of similar degree $d = \deg(\mat{A})$.
However, inspired from other contexts such as approximant and kernel basis
computations~\cite{Storjohann06,ZL2012,JeNeScVi16,za2012} as well as polynomial
matrix inversion~\cite{ZhLaSt15} and the determinant algorithm in this paper,
one may hope for algorithms that are even faster than
$\softO{\matdim^{\expmatmul} d}$ when the degrees in $\mat{A}$ are non-uniform,
for example, if all high-degree entries are located in a few rows and columns
of $\mat{A}$. In the present paper we use ideas in~\cite{GSSV2012} to reduce
the non-uniformity of the degrees in $\mat{A}$ in the context of Hermite form
computation, thus obtaining \cref{thm:hermite}. 

\begin{theorem}
  \label{thm:hermite}
  Let $\mat{A}$ be a nonsingular matrix in $\polMatSpace$. There is a
  deterministic algorithm which computes the Hermite form of $\mat{A}$ using
  $\softO{\matdim^\expmatmul \lceil \degDet/\matdim \rceil} \subseteq
  \softO{\matdim^\expmatmul \lceil s \rceil}$ operations in $\field$, with $s$
  being the minimum of the average of the degrees of the columns of $\mat{A}$
  and that of its rows.
\end{theorem}

The remainder of this paper is organized as follows. In
\Cref{sec:preliminaries} we give preliminary information on shifted degrees as
well as kernel and column bases of polynomial matrices. We also recall why it
is interesting to have cost bounds involving the generic determinant bound
rather than the degree of the matrix; see in particular \cref{rk:degdet}.  \Cref{sec:diagonals} contains the fast algorithm
for finding the diagonal entries of a triangular form.  This is followed in \Cref{sec:determinant} 
by our algorithm for
finding the determinant. The reduction of degrees of off diagonal entries in the Hermite form  is then given in \Cref{sec:hermite}. It computes the remaining
entries by relying in particular on fast deterministic column reduction. In 
\Cref{sec:parlin} we then give the details about how to use partial linearization to
decrease the non-uniformity of the degrees in the input matrix for Hermite form
computation.  The paper ends with a conclusion and
topics for future research.

\section{Preliminaries}
\label{sec:preliminaries}

In this section we first give the basic notations for \emph{column degrees} and
\emph{shifted degrees} of vectors and matrices of polynomials.  We then present
the building blocks used in our algorithms, namely the concepts of {\em kernel
basis} and {\em column basis} for a matrix of polynomials. Finally, we explain
our interest in having cost bounds involving the so-called \emph{generic
determinant bound}.

\subsection{Shifted Degrees}

Our methods make use of the concept of {\em shifted} degrees of polynomial
matrices \cite{BLV:1999}, basically shifting the importance of the degrees in
some of the rows of a basis. For a column vector $\row{p}= \trsp{[p_1,\dots,p_\matdim]}$ 
of univariate polynomials over a field $\field$, its column degree, denoted by $\cdeg{\row{p}}$, 
is the maximum of the degrees of the entries of $\row{p}$, that is, 
\[
  \cdeg{\row{p}} = \max_{1\le i\le \matdim} \deg(p_{i}).
\]
The \emph{shifted column degree} generalizes this standard column degree by
taking the maximum after shifting the degrees by a given integer vector that is
known as a \emph{shift}. More specifically, the shifted column degree of
$\row{p}$ with respect to a shift $\shifts = (\shift{1},\dots,\shift{n}) \in
\shiftSpace$, or the \emph{$\shifts$-column degree} of $\row{p}$, is 
\[
  \cdeg[\shifts]{\row{p}} = \max_{1\le i\le \matdim}(\deg(p_{i})+\shift{i}) =
  \deg(\shiftMat{\shifts} \cdot\mathbf{p}),
\]
where 
\[
\shiftMat{\shifts} =
\diag\left(\var^{\shift{1}},\var^{\shift{2}},\ldots,\var^{\shift{n}}\right) \,
.
\]
For a matrix $\mat{P}$, we use $\cdeg{\mat{P}}$ and $\cdeg[\shifts]{\mat{P}}$
to denote respectively the list of its column degrees and the list of its
shifted $\shifts$-column degrees. For the \emph{uniform shift}
$\shifts=(0,\ldots,0)$, the shifted column degree specializes to the standard
column degree. Similarly, $\cdeg[-\shifts]{\mathbf{P}} \leq 0$ is equivalent to
$\deg(p_{ij}) \leq s_i$ for all $i$ and $j$, that is, $\shifts$ bounds the row
degrees of $\mathbf{P}$.

The shifted row degree of a row vector $\row{q} = [q_{1},\dots,q_{\matdim}]$ is
defined similarly as 
\[
  \rdeg[\shifts]{\mat{q}} = \max_{1\le i\le n}[\deg(q_{i})+\shift{i}] =
  \deg(\mathbf{q}\cdot \shiftMat{\shifts}).
\]
Shifted degrees have been used previously in polynomial matrix computations and
in generalizations of some matrix normal forms \cite{BLV:jsc06}. The shifted
column degree is equivalent to the notion of \emph{defect} commonly used in the
rational approximation literature.

Along with shifted degrees we also make use of the notion of a polynomial
matrix being column reduced. A full-rank polynomial matrix $\mat{A} =
[a_{ij}]_{i,j}$ is column reduced if its leading column coefficient matrix,
that is the matrix
\[
  \leadingMat{\mat{A}} = [\mathrm{coeff}( a_{ij}, x, d_j ) ]_{1\leq i, j \leq
  \matdim}, \text{ with } (d_1,\ldots,d_\matdim) = \cdeg{\mat{A}},
\]
has full rank. Then, the polynomial matrix $\mat{A}$ is $\shifts$-column
reduced if $\shiftMat{\shifts} \mat{A}$ is column reduced. The concept of
$\mat{A}$ being shifted row reduced is similar.

The usefulness of the shifted degrees can be seen from their applications in
polynomial matrix computation problems such as Hermite-Pad\'e and M-Pad\'e
approximations \cite{Beckermann92,BarBul92,BeLa94,ZL2012}, minimal kernel bases
\cite{za2012}, and shifted column reduction~\cite{BLV:jsc06,Neiger16}.

An essential fact needed in this paper, also based on the use of shifted
degrees, is the efficient multiplication of matrices with unbalanced degrees
\cite[Theorem 3.7]{za2012}. 

\begin{theorem}\label{thm:multiplyUnbalancedMatrices} 
  Let $\mat{A} \in \polMatSpace[m][n]$ with $m\le n$, $\shifts \in \N^\matdim$
  a shift with entries bounding the column degrees of $\mat{A}$, and $\xi$ a
  bound on the sum of the entries of $\shifts$. Let $\mat{B} \in
  \polMatSpace[n][k]$ with $k\in \bigO{m}$ and the sum $\theta$ of its
  $\shifts$-column degrees satisfying $\theta\in \bigO{\xi}$. Then we can
  multiply $\mat{A}$ and $\mat{B}$ with a cost of $\softO{ \matdim^{2}
  m^{\expmatmul-2} \lceil s \rceil } \subseteq \softO{ n^{\expmatmul} \lceil s
\rceil }$, where $s=\xi/n$ is the average of the entries of $\shifts$. 
\end{theorem}

\subsection{Shifted Kernel  and Column Bases}

The kernel of $\mat{A} \in \polMatSpace[m][n]$ is the $\polRing$-module
$\{\col{p} \in \polMatSpace[n][1] \;\;|\;\; \mat{A} \col{p} = 0\}$. Such a
module is free and of rank $k \le n$ \cite[Chapter~12, Theorem~4]{DumFoo04};
any of its bases is called a kernel basis of $\mat{A}$. In other words:

\begin{definition}
\label{def:kernelBasis}
Given $\mat{A} \in \polMatSpace[m][n]$, a polynomial matrix $\mat{N}\in
\polMatSpace[n][k]$ is a (right) kernel basis of $\mat{A}$ if the following
properties hold: 
\begin{enumerate}
  \item $\mat{N}$ has full rank, 
  \item $\mat{N}$ satisfies $\mat{A} \cdot \mat{N} = 0$,
  \item Any $\col{q} \in \polMatSpace[n][1]$ satisfying $\mat{A} \col{q} = 0$
    can be written as a linear combination of the columns of $\mat{N}$, that
    is, there exists $\col{p} \in \polMatSpace[k][1]$ such that
    $\col{q}=\mat{N}\col{p}$. 
\end{enumerate}
\end{definition}
\noindent
It is easy to show that any pair of kernel bases $\mat{N}$ and $\mat{M}$ of
$\mat{A}$ are unimodularly equivalent. An $\shifts$-minimal kernel basis of
$\mat{A}$ is a kernel basis that is $\shifts$-column reduced. 

\begin{definition}
  Given $\mat{A} \in \polMatSpace[m][n]$, a matrix $\mat{N} \in
  \polMatSpace[n][k]$ is an $\shifts$-minimal (right) kernel basis of $\mat{A}$
  if $\mat{N}$ is a kernel basis of $\mat{A}$ and $\mat{N}$ is $\shifts$-column
  reduced.
\end{definition}

A column basis of $\mat{A}$ is a basis of the $\polRing$-module $\{\mat{A}
\col{p} , \;\; \col{p} \in \polMatSpace[n][1]\}$, which is free of rank $r \le
n$. Such a basis can be represented as a full rank matrix $\mat{M} \in
\polMatSpace[m][r]$ whose columns are the basis elements. A column basis is not
unique and indeed any column basis right multiplied by a unimodular matrix
gives another column basis.

\begin{example}
  \label{eg:ctd_kernel_column_bases}
  Let 
  \[
    \mat{A} =
    \begin{bmatrix}
                      6x + 1 &    2x^3 + x^2 + 6x + 1 &                      3 \\
      4x^5 + 5x^4 + 4x^2 + x & 6x^5 + 5x^4 + 2x^3 + 4 & x^4 + 5x^3 + 6x^2 + 5x
    \end{bmatrix}
  \]
  be a $2 \times 3$ matrix over $\Z_7[x]$ having column degree $\shifts =
  (5,5,4)$.  Then a column basis $\mat{B}$, and a kernel basis $\mat{N}$, of
  $\mat{A}$ are given by
  \[
    \mat{B} =
    \begin{bmatrix}
      5x + 5 & 1 \\
           3 & 1
    \end{bmatrix}
\quad\text{and}\quad
    \mat{N} =
    \begin{bmatrix}
           6x^6 + 4x^5 + 5x^4 + 3x^3 + 4x^2 + 1 \\
                         4x^4 + 5x^3 + x^2 + 6x \\
      4x^7 + 4x^6 + 4x^5 + 4x^3 + 5x^2 + 3x + 2
    \end{bmatrix}.
  \]
  For example, if $\col{b}_1$ and $\col{b}_2$ denote the columns of $\mat{B}$
  then the third column of $\mat{A}$, denoted by $\col{a}_3$, is given by
  \[
    \col{a}_3  =  (4x^3 + 3x^2 + 6x + 5) \,\col{b}_1 \,+\, (x^4 + 4x^2 + x + 6) \,\col{b}_2 .
  \]
  Here $\cdeg[\shifts]{\mat{N}} = (11)$. In addition, the shifted leading coefficient matrix
  \[
  \leadingMat[\shifts]{\mat{N}} =
  \begin{bmatrix}
      6 \\
      0 \\
      4
  \end{bmatrix}
  \]
  has full rank, and hence we have that $\mat{N}$ is an $\shifts$-minimal
  kernel basis of $\mat{A}$. \qedexmp
\end{example}

Fast algorithms for kernel basis computation and column basis computation are
given in \cite{za2012} and in \cite{za2013}, respectively. In both cases they
make use of fast methods for order bases (often also referred to as minimal
approximant bases) \cite{BeLa94,Giorgi2003,ZL2009,ZL2012}. In what follows, we write
$\sumVec{\shifts}$ for the sum of the entries of a tuple $\shifts \in
\N^\matdim$ with nonnegative entries.

\begin{theorem}\label{thm:fastkercolbasis} 
  Let $\mat{A} \in \polMatSpace[m][n]$ with $m \le n$ and $m \in \Theta(n)$,
  and let $\shifts\in\N^\matdim$ be such that $\cdeg{\mat{A}} \le \shifts$
  componentwise. Then, there exist deterministic algorithms which compute
  \begin{itemize}
    \item[(i)] an $\shifts$-minimal kernel basis of $\mat{A}$ using
      $\softO{\matdim^{\expmatmul} \lceil s\rceil }$ field operations,
    \item[(ii)] a column basis of $\mat{A}$ using $\softO{ \matdim^{\expmatmul}
    \lceil s \rceil }$ field operations,
  \end{itemize}
    where $s=\sumVec{\shifts}/n$ is the average column degree of $\mat{A}$.
\end{theorem}

\subsection{The generic determinant degree bound}
\label{subsec:degdet}

For a nonsingular $n \times n$ matrix $\mat{A} \in \polMatSpace$, the degree of
the determinant of $\mat{A}$ provides a good measure of the size of the output
$\hermite$ in the case of Hermite form computation. Indeed, if we denote by
$\minDegs = (\minDeg_1,\ldots,\minDeg_\matdim)$ the degrees of the diagonal
entries of $\hermite$, then we have $\deg(\det(\mat{A})) = \minDeg_1 + \cdots +
\minDeg_\matdim.$ Since the diagonal entries are those of largest degree in their respective
rows, we directly obtain that $\hermite$ can be represented using $\matdim^2 +
\matdim \sumVec{\minDegs} = \matdim^2 + \matdim \deg(\det(\mat{A}))$ field
elements.

The size of the input $\mat{A}$ can be measured by several quantities, which
differ in how precisely they account for the distribution of the degrees in
$\mat{A}$. It is interesting to relate these quantities to the degree of the
determinant of $\mat{A}$, since the latter measures the size of the output
$\hermite$. A first, coarse bound is given by the maximum degree of the entries
of the matrix: $\mat{A}$ can be represented by $\matdim^2 + \matdim^2
\deg(\mat{A})$ field elements. On the other hand, by definition of the
determinant we have that $\det(\mat{A})$ has degree at most $\matdim
\deg(\mat{A})$. A second, finer bound can be obtained using  the \emph{average}
of the row degrees and of the column degrees: the size of $\mat{A}$ in terms of
field elements is at most $\matdim^2 + \matdim \min(\sumVec{\rdeg{\mat{A}}},
\sumVec{\cdeg{\mat{A}}})$. Again we have the related bound \[
  \deg(\det(\mat{A})) \le \min(\sumVec{\rdeg{\mat{A}}},
  \sumVec{\cdeg{\mat{A}}}).
\] 

An even finer bound on the size of $\mat{A}$ is given by the {\em generic
determinant bound}, introduced
in~\cite[Section~6]{GSSV2012}.
For $\mat{A} =[a_{ij}] \in \polMatSpace[\matdim]$,  this is
defined as
\begin{equation}
  \label{eqn:degDet}
  \degDet = \max_{\pi \in S_\matdim} \sum_{1 \le i\le \matdim} \overline{\deg}(a_{i,\pi_i})
\end{equation}
where $S_\matdim$ is the set of permutations of $\{1,\ldots,\matdim\}$, and where
\[
\overline{\deg}(p) = \left\{ \begin{array}{crr} 0 & \mbox{if} & p = 0 \\
                     \deg(p) & \mbox{if} & p \ne 0 \end{array} \right. .
                 \]
By definition, we have the inequalities 
\[\deg(\det(\mat{A})) \le \degDet \le \min(\sumVec{\rdeg{\mat{A}}},
\sumVec{\cdeg{\mat{A}}}) \le \matdim \deg(\mat{A}),\] and it is easily checked
that $\mat{A}$ can be represented using $\matdim^2 + 2 \matdim \degDet$ field
elements.

Thus in Hermite form computation both the input and the output have
average degree in $\bigO{\degDet/\matdim}$ and can be represented using
$\bigO{\matdim^2 \lceil\degDet/\matdim\rceil}$ field elements. Furthermore
$\degDet$ gives a more precise account of the degrees in $\mat{A}$ than the
average row and column degrees, and an algorithm with cost bound
$\softO{\matdim^\expmatmul \lceil \degDet/\matdim \rceil}$ is always faster,
sometimes significantly, than an algorithm with cost bound
$\softO{\matdim^\expmatmul \lceil s \rceil}$ where $s$ is the average column
degree or the average row degree, let alone $s = \deg(\mat{A})$.

\begin{remark}
  \label{rk:degdet}
Let us justify why this can sometimes be \emph{significantly} faster. We have
seen that $\degDet / \matdim$ is bounded from above by both the average column
degree and the average row degree of $\mat{A}$. It turns out that, in some important cases
$\degDet/\matdim$ may be substantially smaller than these averages. For
example, consider $\mat{A}$ with one row and one column of uniformly large
degree $d$ and all other entries of degree $0$:
\[
\mat{A} =
\begin{bmatrix}
  [d]    & [d] & \cdots & [d] \\
  [d]    & [0] & \cdots & [0] \\
  \vdots & \vdots & \ddots & \vdots \\
  [d]    & [0]    & \cdots & [0]
\end{bmatrix}
\;\; \in \polMatSpace .
\]
Here, the average row degree and the average column degree are both
exactly $d$ while the generic determinant bound is $d$ as well. Thus, here
$\degDet/\matdim = d/\matdim$ is much smaller than $d = \deg(\mat{A}) =
\min(\sumVec{\rdeg{\mat{A}}}/\matdim, \sumVec{\cdeg{\mat{A}}}/\matdim)$.
For similar examples, we refer the reader to \cite[Example~4]{GSSV2012} and
\cite[equation~(8)]{ZhLaSt15}.
\qedexmp
\end{remark}

\section{Determining the diagonal entries of a triangular form}
\label{sec:diagonals}

In this section we show how to determine the diagonal entries of a triangular
form of a nonsingular matrix $\mat{A} \in \polMatSpace$ with $\mat{A}$ having
column degrees $\shifts$. Our algorithm makes use of fast kernel and column
bases computations.

As mentioned in the introduction, we consider unimodularly transforming $\mat{A}$ to 
\begin{equation}
\label{eq:block_triangularization}
\mat{A} \mat{U}=\mat{B}=
\begin{bmatrix}\mat{B}_{1} & \mat{0}\\
\anyMat & \mat{B}_{2}
\end{bmatrix}
\end{equation}
which eliminates a top right block and gives two square diagonal blocks
$\mat{B}_{1}$ and $\mat{B}_{2}$ in $\mat{B}$. After this block
triangularization step, the matrix is now closer to being in triangular form.
Applying this procedure recursively to $\mat{B}_{1}$ and $\mat{B}_{2}$, until
the matrices reach dimension $1$, gives the diagonal entries of a triangular
form of $\mat{A}$. These entries are unique up to multiplication by a nonzero
constant from $\field$, and in particular making them monic yields the diagonal
entries of the Hermite form of~$\mat{A}$.

In this procedure, a major problem is that the degrees in the unimodular
multiplier $\mat{U}$ can be too large for efficient computation.  For example,
the matrix
  \[ \mat{A} = \begin{bmatrix}
      1 & 0 & 0 & \cdots & 0 \\
      -\var^d & 1 & 0 & \cdots & 0 \\
      0 & -\var^d & 1 & \cdots & 0 \\
      \vdots & \ddots & \ddots & \ddots & 0 \\
      0 & \cdots & 0 & -\var^d & 1
  \end{bmatrix} \in \polMatSpace \]
of degree $d>0$ is unimodular and hence its Hermite form is the identity. However the
corresponding unimodular multiplier is
  \[ \mat{U} =  \begin{bmatrix}
      1 & 0 & 0 & \cdots & 0 \\
      \var^d & 1 & 0 & \cdots & 0 \\
      \var^{2d} & \var^d & 1 & \cdots & 0 \\
      \vdots & \ddots & \ddots  & \ddots & 0 \\
      \var^{(\matdim-1)d} & \cdots & \var^{2d} & \var^d & 1
  \end{bmatrix}, ~~~~~~~~~~~~~~~\]
with the sum of the degrees in $\mat{U}$ being in $\Theta(\matdim^3 d)$, beyond
our target cost $\bigO{\matdim^\expmatmul d}$.

\subsection{Fast block elimination}

Our approach is to make use of fast kernel and column basis methods to
efficiently compute the diagonal blocks $\mat{B}_{1}$ and $\mat{B}_{2}$ while
at the same time avoiding the computation of all of $\mat{U}$.

Partition $\mat{A}=\begin{bmatrix}\mat{A}_{u}\\\mat{A}_{d} \end{bmatrix}$, with
$\mat{A}_{u}$ and $\mat{A}_{d}$ consisting of the upper $\lceil n/2\rceil$ and
lower $\left\lfloor n/2\right\rfloor $ rows of $\mat{A}$, respectively.  Then
both upper and lower parts have full-rank since $\mat{A}$ is assumed to be
nonsingular. By partitioning $\mat{U}=\begin{bmatrix}\mat{U}_{\ell} &
  \mat{U}_{r}\end{bmatrix}$, where the column dimension of $\mat{U}_{\ell}$
matches the row dimension of $\mat{A}_{u}$, then $\mat{A} \cdot \mat{U}=
\mat{B}$ becomes
\[
\begin{bmatrix}\mat{A}_{u}\\
\mat{A}_{d}
\end{bmatrix}
\begin{bmatrix}\mathbf{U}_{\ell} & \mathbf{U}_{r}\end{bmatrix}
=
\begin{bmatrix}\mat{B}_{1} & \mat{0} \\
\anyMat & \mat{B}_{2}
\end{bmatrix}.
\]
Notice that the matrix $\mat{B}_{1}$ is nonsingular and is therefore a column
basis of $\mat{A}_{u}$. As such this can be efficiently computed as mentioned
in \Cref{thm:fastkercolbasis}. In order to compute
$\mat{B}_{2}=\mat{A}_{d}\mat{U}_{r}$, notice that the matrix $\mat{U}_{r}$ is a
right kernel basis of $\mat{A}_u$, which makes the top right block of $\mat{B}$
zero. 

The following lemma states that the kernel basis $\mat{U}_{r}$ can be replaced
by any other kernel basis of $\mat{A}_u$ thus giving another unimodular matrix
that also works. 
\begin{lemma}
\label{lem:oneStepHermiteDiagonal}
Partition $\mat{A} = \begin{bmatrix}\mat{A}_{u}\\ \mat{A}_{d} \end{bmatrix}$
and suppose $\mat{B}_{1}$ is a column basis of $\mat{A}_{u}$ and $\mat{N}$ a
kernel basis of $\mat{A}_{u}$. Then there is a unimodular matrix
$\mat{U}=\begin{bmatrix} \anyMat & \mat{N}\end{bmatrix}$ such that
\[
  \mat{A}\mat{U} = \begin{bmatrix}\mat{B}_{1} & \mat{0} \\
  \anyMat & \mat{B}_{2}
\end{bmatrix},
\]
where $\mat{B}_{2}=\mat{A}_{d}\mat{N}$. If $\mat{A}$ is square and nonsingular,
then $\mat{B}_{1}$ and $\mat{B}_{2}$ are also square and nonsingular.
\end{lemma}
\begin{proof}
  This follows from \cite[Lemma~3.1]{za2013}.
\end{proof}

Note that we do not compute the blocks represented by the symbol $\anyMat$.
Thus \Cref{lem:oneStepHermiteDiagonal} allows us to determine $\mat{B}_{1}$ and
$\mat{B}_{2}$ independently without computing the unimodular matrix. This
procedure for computing the diagonal entries is presented in
\Cref{alg:hermiteDiagonal}. Formally the cost of this algorithm is given in
\Cref{prop:diagonalCost}.

\begin{algorithm}[t]
\caption{$\hermiteDiagonal(\mat{A})$}
\label{alg:hermiteDiagonal}

\begin{algorithmic}[1]
\REQUIRE{$\mat{A} \in \polMatSpace$ nonsingular.
}

\ENSURE{$\row{d} \in \polRing^\matdim$ the list of diagonal entries of the
Hermite normal form of $\mat{A}$.}

\IF{$n=1$}
\STATE{write $\mat{A} = \lambda \row{d}$ with $\lambda\in\field$ and $\row{d} \in\polRing$ monic;}
\RETURN{$\row{d}$;}
\ENDIF

\STATE{Partition $\mat{A} := \begin{bmatrix}\mat{A}_{u}\\ \mat{A}_{d}
\end{bmatrix}$, where $\mat{A}_{u}$ consists of the top $\left\lceil
n/2\right\rceil $ rows of $\mat{A}$;}

\STATE{$\mat{B}_{1}:=\colBasis(\mat{A}_{u})$;}

\STATE{$\mat{N}:=\mkb(\mat{A}_{u},\cdeg{\mat{A}})$;}

\STATE{$\mat{B}_{2}:=\mat{A}_{d}\mat{N}$;}

\STATE{$\row{d}_{1}:=\hermiteDiagonal(\mat{B}_{1})$;}
\STATE{$\row{d}_{2}:=\hermiteDiagonal(\mat{B}_{2})$;}

\RETURN{$\left[\row{d}_{1},\row{d}_{2}\right]$;}
\end{algorithmic}
\end{algorithm}

\subsection{Computational cost and example}

Before giving a cost bound for our algorithm, let us observe its correctness on
an example.

\begin{example}
  \label{eg:ctd_hermitediagonal}
Let
\[
  \mat{A} =
  \begin{bmatrix}
                    6x + 1 &       2x^3 + x^2 + 6x + 1 &                      3 \\
    4x^5 + 5x^4 + 4x^2 + x &    6x^5 + 5x^4 + 2x^3 + 4 & x^4 + 5x^3 + 6x^2 + 5x \\
                         2 & 2x^5 + 5x^4 + 5x^3 + 6x^2 &                      6
  \end{bmatrix} ,
\]
working over $\mathbb{Z}_7[x]$. Considering the matrix $\mat{A}_u$ formed by
the top two rows of $\mat{A}$, then a column basis $\mat{B}_1$ and kernel
basis $\mat{N}$ of $\mat{A}_u$ were given in \Cref{eg:ctd_kernel_column_bases}.
If $\mat{A}_d$ denotes the bottom row of $\mat{A}$, then this gives diagonal
blocks
\[
\mat{B}_1 =
\begin{bmatrix}
  5x + 5 & 1 \\
       3 & 1
\end{bmatrix} 
\]
and
\[\mat{B}_2 = \mat{A}_d \mat{N} =
\begin{bmatrix}
  x^9 + 2x^8 + x^7 + 4x^6 + 6x^5 + 4x^4 + 3x^3 + 3x^2 + 4x
\end{bmatrix} .
\]
Recursively computing with $\mat{B}_1$, we obtain a column basis and kernel
basis of the top row $\mat{B}_{1,u}$ of $\mat{B}_1$, as
\[
\tilde{\mat{B}}_1 =
\begin{bmatrix}
  1
 \end{bmatrix}
 \quad\text{and}\quad
  \tilde{\mat{N}} =
  \begin{bmatrix} 
          1 \\
     2x + 2
  \end{bmatrix} .
\]
If $\mat{B}_{1,d}$ denote the bottom row of $\mat{B}_1$, we get
$\tilde{\mat{B}}_2 = \mat{B}_{1,d} \, \tilde{\mat{N}} = \begin{bmatrix} 2x + 5
\end{bmatrix}$, which gives the second diagonal block from $\mat{B}_1$. Thus we
have the diagonal entries of a triangular form of $\mat{B}_1$. On the other
hand, since $\mat{B}_2$ is already a $1\times 1$ matrix we do not need to do
any extra work. As a result we have that $\mat{A}$ is unimodularly equivalent
to 
\[
  \begin{bmatrix}
    1 &       \\
    * & 2x+5     \\
    * &    * & x^9 + 2x^8 + x^7 + 4x^6 + 6x^5 + 4x^4 + 3x^3 + 3x^2 + 4x
  \end{bmatrix},
\]
giving, up to making them monic, the diagonal entries of the Hermite form of
$\mat{A}$. \qedexmp
\end{example}

\begin{proposition}\label{prop:diagonalCost}
\Cref{alg:hermiteDiagonal} costs $\softO{\matdim^{\expmatmul} \lceil s \rceil}$
field operations to compute the diagonal entries of the Hermite normal form of
a nonsingular matrix $\mat{A} \in \polMatSpace$, where $s
=\sumVec{\cdeg{\mat{A}}}/\matdim$ is the average column degree of $\mat{A}$. 
\end{proposition}
\begin{proof}
The three main operations are computing a column basis of $\mat{A}_{u}$,
computing a kernel basis $\mat{N}$ of $\mat{A}_{u}$, and multiplying the
matrices $\mat{A}_{d}\mat{N}$. Let $\shifts$ denote the column degree of
$\mat{A}$ and set $\xi = \sumVec{\shifts}$, an integer used to measure size for
our problem.

For the column basis computation, by \Cref{thm:fastkercolbasis} (see also
\cite[Theorem~5.6]{za2013}) we know that a column basis $\mat{B}_{1}$ of
$\mat{A}_{u}$ can be computed with a cost of $\softO{\matdim^{\expmatmul}
\lceil s \rceil }$, where $s = \xi/\matdim$. Furthermore, the sum of the column
degrees of the computed $\mat{B}_{1}$ is bounded by the sum of the column
degrees of $\mat{A}_u$ (see \cite{za2013}, in particular the proof of Lemma~5.5
therein). Thus, since $\cdeg{\mat{A}_u} \le \shifts$ componentwise, the sum of
the column degrees of $\mat{B}_1$ is at most $\xi$.

Similarly, according to \Cref{thm:fastkercolbasis} (see also
\cite[Theorem~4.1]{za2012}), computing an $\shifts$-minimal kernel basis
$\mat{N}$ of $\mat{A}_{u}$ costs $\softO{\matdim^\expmatmul \lceil s \rceil}$
operations, and the sum of the $\shifts$-column degrees of the output kernel
basis $\mat{N}$ is bounded by $\xi$~\cite[Theorem~3.4]{za2012}.
 
For the matrix multiplication $\mat{A}_{d}\mat{N}$, we have that the sum of the
column degrees of $\mat{A}_{d}$ and the sum of the $\shifts$-column degrees of
$\mat{N}$ are both bounded by $\xi$. Therefore
\Cref{thm:multiplyUnbalancedMatrices} applies and the multiplication can be
done with a cost of $\softO{\matdim^{\expmatmul} \lceil s \rceil}$.  Furthermore, 
since the entries of $\vec s$ bounds the corresponding column degrees of $\mat{A}_d$, 
according to \cite[Lemma~3.1]{za2012}, we have that
the column degrees of $\mat{B}_2 = \mat{A}_{d}\mat{N}$ are bounded by the $\shifts$-column degrees 
of $\mat{N}$. In particular, the sum of the column degrees of $\mat{B}_2$ is at most $\xi$.

If we let the cost of \Cref{alg:hermiteDiagonal} be $g(\matdim)$ for an input
matrix of dimension $\matdim$ then
\begin{eqnarray*}
  g(\matdim) & \in & \softO{\matdim^{\expmatmul} \lceil s \rceil} +
g(\lceil \matdim/2\rceil)+g(\lfloor \matdim/2\rfloor ).
\end{eqnarray*}
As $s = \xi / \matdim$ depends on $\matdim$ we use $\softO{\matdim^{\expmatmul}
\lceil s \rceil} = \softO{\matdim^{\expmatmul} ( s + 1) } =
\softO{\matdim^{\expmatmul-1} \xi + \matdim^\expmatmul}$ with $\xi$ not
depending on $\matdim$. Then we solve the recurrence relation as
\begin{align*}
  g(\matdim) & \in \softO{\matdim^{\expmatmul-1} \xi + \matdim^\expmatmul}
  + g(\lceil \matdim/2\rceil)+g(\lfloor \matdim/2\rfloor )\\
  & \subseteq \softO{\matdim^{\expmatmul-1} \xi + \matdim^\expmatmul} + 2 g(\lceil \matdim/2\rceil )\\
  & \subseteq \softO{\matdim^{\expmatmul-1} \xi + \matdim^\expmatmul} = \softO{\matdim^{\expmatmul} \lceil s \rceil}.
  \qedhere
\end{align*}
\end{proof}

In this cost bound, we do not detail the logarithmic factors because it is not
clear to us for the moment how many logarithmic factors arise from the calls to
the kernel basis and column basis algorithms of \cite{za2012,za2013}, where
they are not reported. Yet, from the recurrence relation above, it can be
observed that no extra logarithmic factor will be introduced if $\expmatmul>2$,
while an extra factor logarithmic in $\matdim$ will be introduced if
$\expmatmul=2$.

\section{Efficient Determinant Computation}
\label{sec:determinant}

In this section, we show how to recursively and efficiently compute the
determinant of a nonsingular matrix $\mat{A}\in\polMatSpace$ having column
degrees $\shifts$. Our algorithm follows a strategy similar to the recursive
block triangularization in \Cref{sec:diagonals}, making use of fast kernel
basis and column basis computation.

Indeed, after unimodularly transforming $\mat{A}$ to 
\[\mat{A} \mat{U} = \mat{B} =
\begin{bmatrix}
    \mat{B}_1 & \mat{0} \\
    \anyMat & \mat{B}_2
\end{bmatrix}\]
as in \cref{eq:block_triangularization}, the determinant
of $\mat{A}$ can be computed as 
\begin{equation}
  \label{eq:determinantFromDiagonalBlocks}
  \det(\mat{A})=\frac{\det(\mat{B})}{\det(\mat{U})}=
  \frac{\det(\mat{B}_{1})\det(\mat{B}_{2})}{\det(\mat{U})},
\end{equation}
which requires us to first compute $\det(\mat{B}_{1})$, $\det(\mat{B}_{2})$,
and $\det(\mat{U})$. The same procedure can then be applied to compute the
determinant of $\mat{B}_{1}$ and the determinant of $\mat{B}_{2}$. However, as 
$\mat{U}$ is unimodular we will handle its determinant differently. This can be
repeated recursively until the dimension becomes $1$.

One major obstacle for efficiency of this approach is that we do want to
compute the scalar $\det(\mat{U})$, and as noted in \Cref{sec:diagonals}, the
degrees of the unimodular matrix $\mat{U}$ can be too large for efficient
computation.
To sidestep this issue, we will show that $\det(\mat{U})$ can be computed with
only partial knowledge of the matrix $\mat{U}$. Combining this with the method
of \Cref{sec:diagonals} to compute the matrices $\mat{B}_{1}$ and $\mat{B}_{2}$
without computing all of $\mat{B}$ and $\mat{U}$, we obtain an efficient
recursive algorithm.

\begin{remark}
  \label{rk:det_from_diagonal}
  In some cases, the computation of the determinant is easily done from the
  diagonal entries of a triangular form. Indeed, let $\mat{A} \in \polMatSpace$
  be nonsingular and assume that we have computed the diagonal entries
  $h_{11},\ldots,h_{\matdim \matdim}$ of its Hermite form. Then, $\det(\mat{A})
  = \lambda h_{11} \cdots h_{\matdim\matdim}$ for some nonzero constant
  $\lambda\in\field$. If the constant coefficient of $h_{11} \cdots
  h_{\matdim\matdim}$ is nonzero, we can retrieve $\lambda$ by computing the
  constant coefficient of $\det(\mat{A})$, which is found by $\field$-linear
  algebra using $\bigO{\matdim^\expmatmul}$ operations since $\det(\mat{A})(0)
  = \det(\mat{A}(0))$. More generally, if we know $\alpha \in\field$ such that
  $h_{11}(\alpha) \cdots h_{\matdim\matdim}(\alpha) \neq 0$, then we can deduce
  $\det(\mat{A})$ efficiently. Yet, this does not lead to a fast deterministic
  algorithm in general since it may happen that $\det(\mat{A})(\alpha)=0$ for
  all field elements $\alpha$, or that finding $\alpha$ with $h_{11}(\alpha)
  \cdots h_{\matdim\matdim}(\alpha) \neq 0$ is a difficult task.
\qedexmp
\end{remark}

We now focus on computing the determinant of $\mat{U}$, or equivalently, the
determinant of $\mat{V}=\mat{U}^{-1}$. The column basis computation from
\cite{za2013} for computing the $m \times m$ diagonal block $\mat{B}_{1}$ also
gives $\mat{U}_{r}$, the matrix consisting of the right $(\matdim-m)$ columns
of $\mat{U}$, which is a right kernel basis of $\mat{A}_{u}$. In fact, this
column basis computation also gives a right factor multiplied with the column
basis $\mat{B}_{1}$ to give $\mat{A}_{u}$. The following lemma shows that this
right factor coincides with the matrix $\mat{V}_{u}$ consisting of the top $m$
rows of $\mat{V}$. The column basis computation therefore gives both
$\mat{U}_{r}$ and $\mat{V}_{u}$ with no additional work.

\begin{lemma}
Let $m$ be the dimension of $\mat{B}_{1}$. The matrix $\mat{V}_{u} \in
\polMatSpace[m][\matdim]$ satisfies $\mat{B}_{1} \mat{V}_{u}=\mat{A}_{u}$ if
and only if $\mat{V}_{u}$ is the submatrix of $\mat{V}=\mat{U}^{-1}$ formed by
its top $m$ rows.
\end{lemma}
\begin{proof}
The proof follows directly from 
\[
\mat{B} \mat{V} =
\begin{bmatrix}
  \mat{B}_{1} & \mat{0}\\
  \anyMat & \mat{B}_{2}
\end{bmatrix}
\begin{bmatrix}
  \mat{V}_{u}\\
  \mat{V}_{d}
\end{bmatrix}
=
\begin{bmatrix}
  \mat{A}_{u}\\
  \mat{A}_{d}
\end{bmatrix}
=\mat{A} \qedhere.
\]
\end{proof}

While the determinant of $\mat{V}$ or the determinant of $\mat{U}$ is needed to
compute the determinant of $\mat{A}$, a major problem is that we do not know
$\mat{U}_{\ell}$ or $\mat{V}_{d}$, which may not be efficiently computed due to
their possibly large degrees. This means we need to compute the determinant of
$\mat{V}$ or $\mat{U}$ without knowing the complete matrix $\mat{V}$ or
$\mat{U}$. The following lemma shows how this can be done using just
$\mat{U}_{r}$ and $\mat{V}_{u}$, which are obtained from the computation of the
column basis $\mat{B}_{1}$.

\begin{lemma}
\label{lem:scalingToDeterminant}
Let $\mat{U} = \begin{bmatrix} \mat{U}_{\ell} & \mat{U}_{r}\end{bmatrix}$ and
$\mat{A}$ satisfy, as before, 
\[
\mat{A} \mat{U}=\begin{bmatrix}
  \mat{A}_{u}\\
  \mat{A}_{d}
\end{bmatrix}
\begin{bmatrix}
  \mat{U}_{\ell} & \mat{U}_{r}
\end{bmatrix}
=\begin{bmatrix}
  \mat{B}_{1} & \mat{0}\\
  \anyMat & \mat{B}_{2}
\end{bmatrix}=\mat{B},
\]
where the row dimension of $\mat{A}_{u}$, the column dimension of
$\mat{U}_{\ell}$, and the dimension of $\mat{B}_{1}$ are $m$.  Let
$\mat{V}=\begin{bmatrix}\mat{V}_{u}\\ \mat{V}_{d}
\end{bmatrix}$ be the inverse of $\mat{U}$ with $m$ rows in $\mat{V}_{u}$ and $\mat{U}_{\ell}^{*} \in \polMatSpace[\matdim][m]$ be a matrix such that
$\mat{U}^{*} = \begin{bmatrix} \mat{U}_{\ell}^{*} & \mat{U}_{r}\end{bmatrix}$
is unimodular. Then $\mat{V}_{u}\mat{U}_{\ell}^{*}$ is unimodular and
\[
\det(\mat{A}) = \frac{\det(\mat{B})\det(\mat{V}_{u}\mat{U}_{\ell}^{*})}{\det(\mat{U}^{*})}.
\]
\end{lemma}
\begin{proof}
Since $\det(\mat{A}) = \det(\mat{B})\det(\mat{V})$, we just need to show that
$\det(\mat{V}) = \det(\mat{V}_{u}\mat{U}_{\ell}^{*})/\det(\mat{U}^{*})$. This
follows from 
\begin{eqnarray*}
\det(\mat{V})\det(\mat{U}^{*}) & = & \det(\mat{V}\mat{U}^{*}) \\
 & = &
\det\left(\begin{bmatrix}
  \mat{V}_{u}\\
  \mat{V}_{d}
\end{bmatrix}
\begin{bmatrix}
  \mat{U}_{\ell}^{*} & \mat{U}_{r}\end{bmatrix}\right)\\
& = & \det\left(\begin{bmatrix}\mat{V}_{u}\mat{U}_{\ell}^{*} & \mat{0} \\
\anyMat & \idMat
\end{bmatrix}\right)\\
 & = & \det(\mat{V}_{u}\mat{U}_{\ell}^{*}).
\end{eqnarray*}
In particular $\det(\mat{V}_{u}\mat{U}_{\ell}^{*})$ is a nonzero constant and
thus $\mat{V}_{u}\mat{U}_{\ell}^{*}$ is unimodular.
\end{proof}

\Cref{lem:scalingToDeterminant} shows that the determinant of $\mat{V}$
can be computed using $\mat{V}_{u}$, $\mat{U}_{r}$, and a unimodular completion
$\mat{U}^*$ of $\mat{U}_{r}$. In fact, this can be made more efficient still by
noticing that since we are looking for a constant determinant, the higher
degree parts of the matrices do not affect the computation. Indeed, if 
$\mat{U}\in \polMatSpace$ is unimodular, then one has
\begin{equation}\label{determinantOfUnimodular}
 \det(\mat{U}) = \det(\mat{U}
\bmod x) = \det(\mat{U}(0))
\end{equation}
since 
\[
\det(\mat{U}\bmod x) = \det(\mat{U}(0)) = \det(\mat{U})(0) = \det(\mat{U})\bmod
x =\det(\mat{U}).
\]

\Cref{determinantOfUnimodular} allows us to use just the degree zero
coefficient matrices in the computation. Hence \Cref{lem:scalingToDeterminant}
can be improved as follows.

\begin{lemma}
\label{lem:scalingToDeterminantSimplified} Let $\mat{A}$,
$\mat{U}= \begin{bmatrix}\mat{U}_{\ell} & \mat{U}_{r} \end{bmatrix}$, and
$\mat{V}=\begin{bmatrix}\mat{V}_{u}\\ \mat{V}_{d}
\end{bmatrix}$ be as before. Let $U_{r}=\mat{U}_{r}\bmod x$ and
$V_{u}=\mat{V}_{u}\bmod x$ be the constant matrices of $\mat{U}_{r}$ and
$\mat{V}_{u}$, respectively. Let $U_{\ell}^{*}\in \matSpace[\matdim][m]$ be a
matrix such that $U^{*} = \begin{bmatrix} U_{\ell}^{*} & U_{r}\end{bmatrix}$ is
nonsingular. Then
\[
\det(\mat{A}) = \frac{\det(\mat{B})\det(V_{u}U_{\ell}^{*})}{\det(U^{*})}.
\]
\end{lemma}
\begin{proof}
Suppose $\mat{U}_\ell^* \in \polMatSpace[\matdim][m]$ is such that $U_\ell^* =
\mat{U}_\ell^* \bmod x$ and $\mat{U}^* = \begin{bmatrix} \mat{U}_{\ell}^{*} &
  \mat{U}_{r} \end{bmatrix}$ is unimodular. Using
\Cref{lem:scalingToDeterminant} and \cref{determinantOfUnimodular}, we have
that $\mat{V}_{u}\mat{U}_{\ell}^{*}$ is unimodular with $V_u U_\ell^* =
\mat{V}_{u}\mat{U}_{\ell}^{*} \bmod x$ and thus \[\det(\mat{A}) = \det(\mat{B})
\det(\mat{V}_{u}\mat{U}_{\ell}^{*}) / \det(\mat{U}^{*}) = \det(\mat{B})
\det(V_u U_\ell^*) / \det(U^{*}).\]

Let us now show how to construct such a matrix  $\mat{U}_\ell^*$. Let $\mat{W}_\ell^* \in
\polMatSpace[\matdim][m]$ be any matrix such that $\mat{W}^* = \begin{bmatrix}
  \mat{W}_{\ell}^{*} & \mat{U}_{r} \end{bmatrix}$ is unimodular and let
$W_\ell^*$ denote its constant term $W_\ell^* = \mat{W}_\ell^* \bmod x$. 
It is easily checked that $$\begin{bmatrix} W_\ell^* & U_r \end{bmatrix}^{-1}
\begin{bmatrix} U_\ell^* & U_r
\end{bmatrix} = \begin{bmatrix} T_u & 0 \\ T_d & I \end{bmatrix}$$ for some
nonsingular $T_u \in \matSpace[m]$ and some $T_d \in \matSpace[\matdim-m][m]$.
Define the matrix $\mat{U}_\ell^* = \mat{W}_\ell^*  \begin{bmatrix} T_u \\ T_d
\end{bmatrix}$ in $\polMatSpace[\matdim][m]$. On the one hand, we have that 
the matrix $\mat{U}^* = \begin{bmatrix} \mat{U}_{\ell}^{*} & \mat{U}_{r}
\end{bmatrix} = \mat{W}^* \begin{bmatrix} T_u & 0 \\ T_d & I \end{bmatrix}$ is
unimodular. On the other hand, by construction we have that $\mat{U}_\ell^*
\bmod x = W_\ell^* \begin{bmatrix} T_u \\ T_d
\end{bmatrix} = U_\ell^*$.
\end{proof}

Thus \cref{lem:scalingToDeterminantSimplified} requires us to compute
$U_{\ell}^{*}\in \matSpace[\matdim][m]$ a matrix such that $U^{*} =
\begin{bmatrix} U_{\ell}^{*} & U_{r}\end{bmatrix}$ is nonsingular. This can be
obtained from the nonsingular matrix that transforms $V_{u}$ to its reduced
column echelon form computed using the Gauss Jordan transform algorithm from
\cite{storjohann:phd2000} with a cost of $\bigO{\matdim m^{\expmatmul-1}}$
field operations.

We now have all the ingredients needed for computing the determinant of
$\mat{A}$. A recursive algorithm is given in
\Cref{alg:determinant}, which computes the determinant of $\mat{A}$ as the
product of the determinant of $\mat{V}$ and the determinant of $\mat{B}$. The
determinant of $\mat{B}$ is computed by recursively computing the determinants
of its diagonal blocks $\mat{B}_{1}$ and $\mat{B}_{2}$.

\begin{algorithm}[t]
\caption{$\determinant(\mat{A})$}
\label{alg:determinant}

\begin{algorithmic}[1]
\REQUIRE{$\mat{A}\in \polMatSpace$, nonsingular.
}

\ENSURE{the determinant of $\mat{A}$.}

\IF{$n=1$}
\RETURN{$\mat{A}$;}
\ENDIF

\STATE{$\begin{bmatrix}\mat{A}_{u}\\
\mat{A}_{d}
\end{bmatrix}:=\mat{A}$, with $\mat{A}_{u}$ consisting of the top $\left\lceil n/2\right\rceil $
rows of $\mat{A}$;}

\STATE{$\mat{B}_{1},\mat{U}_{r},\mat{V}_{u}:=\colBasis(\mat{A}_{u})$; \\
  \hspace{0.15in}
  {\bf Note:} Here $\colBasis()$ also returns the kernel basis $\mat{U}_{r}$ \\
  \hspace{0.15in}
  and the right factor $\mat{V}_{u}$ such that $\mat{A}_u = \mat{B}_{1}
  \mat{V}_u$.}

\STATE{$\mat{B}_{2}:=\mat{A}_{d}\mat{U}_{r}$;}

\STATE{$U_{r}:=\mat{U}_{r}\bmod x$; $V_{u}:=\mat{V}_{u}\bmod x$;}

\STATE{
  Compute a matrix $U_{\ell}^{*} \in \matSpace[\matdim][\lceil \matdim/2\rceil]$ such that
$U^{*}= \begin{bmatrix} U_{\ell}^{*} & U_{r}\end{bmatrix}$ is nonsingular;
}

\STATE{$d_{V} := \det(V_{u}U_{\ell}^{*}) / \det(U^{*})$ ~~ (element of $\field$);}

\STATE{$\mathbf{d}_{\mat{B}}:=\determinant(\mat{B}_{1})\determinant(\mat{B}_{2});$}

\RETURN{$d_{V}\mat{d}_{B}$;}
\end{algorithmic}
\end{algorithm}

\begin{proposition}
\label{prop:det} \Cref{alg:determinant} costs
$\softO{\matdim^{\expmatmul} \lceil s \rceil}$ field operations to compute the
determinant of a nonsingular matrix $\mat{A} \in \polMatSpace$, where $s$ is
the average column degree of $\mat{A}$.
\end{proposition}
\begin{proof}
From \Cref{lem:oneStepHermiteDiagonal} and
\Cref{prop:diagonalCost} the computation of the two diagonal blocks
$\mat{B}_{1}$ and $\mat{B}_{2}$ costs $\softO{\matdim^{\expmatmul}\lceil s
\rceil}$ field operations. As mentioned above, computing $U_l^*$ at Step 6 of
the algorithm costs $\bigO{\matdim^\expmatmul}$ operations. Step 7 involves
only constant matrices so that $d_V$ can be computed
$\bigO{\matdim^\expmatmul}$. Finally, $\det(\mat{B}_{1})$ and
$\det(\mat{B}_{2})$ are computed recursively and multiplied. Since these are
two univariate polynomials of degree at most $\deg(\det(\mat{A})) \le \xi =
\matdim s$, their product $\mathbf{d}_\mat{B}$ is obtained in $\softO{\xi}
\subset \softO{\matdim^\expmatmul \lceil s \rceil}$ operations.

Therefore, the recurrence relation for the cost of the
\Cref{alg:determinant} is the same as that in the proof of
\Cref{prop:diagonalCost}, and the total cost is
$\softO{\matdim^{\expmatmul} \lceil s \rceil}$.
\end{proof}

\Cref{prop:det} can be further improved using the following result
from~\cite[Corollary~3]{GSSV2012}.

\begin{proposition}
  \label{prop:reduction_determinant}
  Let $\mat{A} \in \polMatSpace$ be nonsingular. Using no operation in
  $\field$, one can build a matrix $\hat{\mat{A}} \in
  \polMatSpace[\hat{\matdim}]$ such that
  \begin{enumerate}[(i)]
    \item $\matdim \le \hat{\matdim} < 3 \matdim$ and
     $\deg(\hat{\mat{A}}) \le \lceil \degDet / \matdim \rceil$,
    \item the determinant of $\mat{A}$ is equal to the determinant of
      $\hat{\mat{A}}$.
  \end{enumerate}
\end{proposition}
This reduction, combined with our result in \cref{prop:det}, proves
\Cref{thm:determinant}.

\begin{example}
  \label{eg:ctd_determinant}
In order to observe the correctness of the algorithm, let
\[
  \mat{A}= \begin{bmatrix}
    -x+2 & -2x-3 & 3x^3+x^2 & -x+2    & -3x^5-x^4    \\ 
    -x   &   -2  & 3x^3     & -x      & -3x^5        \\ 
    -2   &   x+3 & 2        & -2      & -2x^2        \\ 
    0    &   1   & -3x^2-2  & -2x^2-1 & x^4+x^2      \\ 
    0    &   2   & 3        & -3x^2   & -2x^4-3x^2+3
  \end{bmatrix}
\]
working over $\Z_7[x]$. If $\mat{A}_u$ denotes the top three rows of
$\mat{A}$, then we have a column basis 
\[
  \mat{B}_1= \begin{bmatrix}
    -x+2 & -2x-3 & 3x^3+x^2 \\ 
    -x   & -2    & 3x^3     \\ 
    -2   & x+3   & 2
  \end{bmatrix}
\]
and a minimal kernel basis 
\[
  \mat{U}_r= \begin{bmatrix}
    3  & 0   \\ 
    0  & 0   \\ 
    0  & x^2 \\ 
    -3 & 0  \\ 
    0  & 1
  \end{bmatrix}
\]
for $\mat{A}_u$. The second block diagonal is then given by
\[
  \mat{A}_d \mat{U}_r = \begin{bmatrix}
    x^2-3 & -2x^4-x^2 \\
    -2x^2 & -2x^4+3
  \end{bmatrix}.
\]
The computation of the column basis $\mat{B}_1$ also gives the right factor 
\[
  \mat{V}_u= \begin{bmatrix}
    1 & 0 & 0 & 1 & 0    \\
    0 & 1 & 0 & 0 & 0    \\
    0 & 0 & 1 & 0 & -x^2
  \end{bmatrix}
\]
and so the constant term matrices are then 
\[
  U_r = \begin{bmatrix}
    3  & 0 \\
    0  & 0 \\
    0  & 0 \\
    -3 & 0 \\
    0  & 1
  \end{bmatrix} \quad\text{and}\quad
  V_u = \begin{bmatrix}
    1 & 0 & 0 & 1 & 0\\
    0 & 1 & 0 & 0 & 0\\
    0 & 0 & 1 & 0 & 0
  \end{bmatrix}
\]
with Gaussian-Jordan elimination used to find a nonsingular completion of $U_r$ as  
\[
  U_{\ell}^{*} = \begin{bmatrix}
    1 & 0 & 0 \\
    0 & 1 & 0 \\
    0 & 0 & 1 \\
    0 & 0 & 0 \\
    0 & 0 & 0
  \end{bmatrix}.
\]
The determinant of $\mat{U}$ is then computed as 
\[
  d_V = \frac{\det(V_u U_\ell^*)}{\det(U^*)} = -\frac{1}{3} = 2
\]
where we recall that $U^* = [U_{\ell}^{*} \;\; U_r]$. The determinants of $\mat{B}_1$ and $\mat{B}_2$ are then computed recursively.
In the case of $\mat{B}_1$ a minimal kernel basis and column basis are given by
\[
  \mat{U}_{r,1} = \begin{bmatrix}
    3 x^2 \\
    0     \\
    1
  \end{bmatrix}  \, ,\;\;\;
  \mat{B}_{1,1} = \begin{bmatrix}
    -x+2 & 0 \\
    -x   & 2x-2
  \end{bmatrix} \, ,\;\;\;\text{and}\;\;
  \mat{V}_{u,1} = 
  \begin{bmatrix}
    1 & 2 & -3x^2 \\
    0 & 1 & 0
  \end{bmatrix}.
\]
This gives the remaining diagonal block as $\mat{B}_{1,2}  = \begin{bmatrix}
  x^2 + 2 \end{bmatrix}$. The corresponding constant term matrices $U_{r,1}$
and $V_{u,1}$ and nonsingular completion $U_{\ell,1}^*$ are then given by
\[
  U_{r,1} = \begin{bmatrix} 0 \\ 0 \\ 1 \end{bmatrix} \, ,\quad
  V_{u,1} = \begin{bmatrix} 1 & 2 & 0 \\ 0 & 1 & 0 \end{bmatrix} \, ,\;\;\;\text{and}\;\;
  U_{\ell,1}^* = \begin{bmatrix} 1 & 0  \\ 0 & 1 \\ 0 & 0 \end{bmatrix},
\]
which gives $d_{V_1} = 1$. Hence $\det(\mat{B}_1) = (-x + 2)(2x-2)(x^2 + 2)$.
A similar argument gives $\det(\mat{B}_2) = (x^2 - 3)(x^4 + 3)$ and hence
\[
\det(\mat{A}) = d_V \det(\mat{B}_1) \det(\mat{B}_2) =
3x^{10} - 2x^9 + 3x^8 + 2x^7 - x^6 - x^5 + x^4 - x^3 - 2x^2 + x - 3.
\qedexmp
\]
\end{example}

\section{Fast computation of the Hermite form}
\label{sec:hermite}

In \Cref{sec:diagonals}, we have shown how to efficiently determine the
diagonal entries of the Hermite normal form of a nonsingular input matrix
$\mat{A} \in \polMatSpace$. One then still needs to determine the remaining
entries for the complete Hermite form $\hermite$ of $\mat{A}$.

Here, we observe that knowing the diagonal degrees of $\hermite$ allows us to
use partial linearization techniques~\cite[Section 6]{GSSV2012} to reduce to
the case of computing a column reduced form of $\mat{A}$ for an almost uniform
shift. Along with the algorithm in \Cref{sec:diagonals}, this gives an
algorithm to compute the Hermite form of $\mat{A}$ in
$\softO{\matdim^\expmatmul \deg(\mat{A})}$ field operations using fast
deterministic column reduction~\cite{GSSV2012}. 

\subsection{Hermite form via shifted column reduction}
\label{subsec:hermite_via_reduction}

It is known that the Hermite form $\hermite$ of $\mat{A}$ is a
\emph{shifted} reduced form of $\mat{A}$ for a whole range of shifts. Without
further information on the degrees in $\hermite$, one appropriate shift is
\begin{equation}\label{eqn:popov-shift}
\shifts[h] \;\;=\;\;   (\matdim(\matdim-1)d, ~~\matdim(\matdim-2)d, ~~~\ldots,~~\matdim
d, ~~ 0)
\end{equation}
where $d = \deg(\mat{A})$ (cf. \cite[Lemma 2.6]{BLV:jsc06}).
Note that this shift has a large amplitude, namely $\max(\shifts[h]) - \min(\shifts[h]) \in
\Theta(\matdim^2 d)$. Unfortunately we are not aware of a deterministic shifted
reduction algorithm that would compute an $\shifts[h]$-reduced form of
$\mat{A}$ in $\softO{\matdim^\expmatmul d}$ field operations.

Now, let us consider the degrees $\minDegs =
(\minDeg_1,\ldots,\minDeg_\matdim)$ of the diagonal entries of $\hermite$. Then
we have that $\hermite$ is a $-\minDegs$-column reduced form of $\mat{A}$ and,
in addition, that $\hermite$ can be easily recovered from any
$-\minDegs$-column reduced form of $\mat{A}$.  More precisely, suppose that we
know $\minDegs$, for example thanks to the algorithm in \Cref{sec:diagonals}.
Then, we claim that $\hermite$ can be computed as follows, where $\shifts[\mu]
= (\max(\minDegs),\ldots,\max(\minDegs)) \in \N^\matdim$:
\[
  \begin{array}{ccccc}
    \shiftMat{\shifts[\mu]-\minDegs}\mat{A}
    & \xrightarrow{\;\;\text{reduction}\;\;}
    & \shiftMat{\shifts[\mu]-\minDegs}\mat{R}
    & \xrightarrow{\;\;\text{normalization}\;\;}
    & \hermite = \mat{R} \, \cdot \leadingMat[{-\minDegs}]{\mat{R}}^{-1}
  \end{array} 
\]
where $\mat{R}$ is any $-\minDegs$-column reduced form of $\mat{A}$.
To show this, we will rely on the following consequence of \cite[Lemma 17]{SarSto11}.

\begin{lemma}
  \label{lem:popov_normalize0}
  Let $\mat{A}$ and $\mat{B}$ be column reduced matrices in
  $\polMatSpace$ with uniform column degree $(d,\ldots,d)$, for some $d \in \N$. 
  If $\mat{A}$ and $\mat{B}$ are right-unimodularly equivalent then $$\mat{A}
  \, \cdot \leadingMat{\mat{A}}^{-1} = \mat{B} \, \cdot
  \leadingMat{\mat{B}}^{-1}.$$
\end{lemma}
\begin{proof}
  The matrix $\mat{A}$ is column reduced with uniform column degree
  $(d,\ldots,d)$. As such $ \mat{A} \,\cdot \leadingMat{\mat{A}}^{-1}$ is its
  Popov form according to \cite[Lemma 17]{SarSto11} (i.e. its leading
  coefficient matrix is the identity). Similarly, $\mat{B} \, \cdot
  \leadingMat{\mat{B}}^{-1}$ is the Popov form of $\mat{B}$ in this case.  We
  recall that the Popov form is a canonical form under right-unimodular
  equivalence for nonsingular matrices in $\polMatSpace$; for a general
  definition we refer the reader to~\cite{kailath:1980}. Thus, since $\mat{A}$
  and $\mat{B}$ are right-unimodularly equivalent, the uniqueness of the Popov
  form implies $\mat{A} \, \cdot \leadingMat{\mat{A}}^{-1} = \mat{B} \, \cdot
  \leadingMat{ \mat{B}}^{-1}$.
\end{proof}

As we often wish to apply \cref{lem:popov_normalize0} with shifts we also
include the following.
\begin{lemma}
  \label{lem:popov_normalize}
  Let $\shifts \in \shiftSpace$ be a shift, and let $\mat{A}$ and $\mat{B}$ be
  $\shifts$-column reduced matrices in $\polMatSpace$ with uniform
  $\shifts$-column degree $(d,\ldots,d)$, for some $d \in \Z$. If $\mat{A}$ and
  $\mat{B}$ are right-unimodularly equivalent then $$\mat{A} \, \cdot
  \leadingMat[\shifts]{\mat{A}}^{-1} = \mat{B} \, \cdot
  \leadingMat[\shifts]{\mat{B}}^{-1}.$$
\end{lemma}
\begin{proof}
We simply replace $\mat{A}$ and $\mat{B}$ by $\shiftMat{\shifts} \,\mat{A}$ and
$\shiftMat{\shifts} \,\mat{B}$ in the previous proof.
\end{proof}

In addition, since the Hermite form of $\mat{A}$ is the shifted Popov form of
$\mat{A}$ for the shift $\shifts[h]$ in \cref{eqn:popov-shift}, we can state the
following specific case of~\cite[Lemma 4.1]{JeNeScVi16}.

\begin{corollary}
  \label{cor:hermite_knowndeg}
  Let $\mat{A} \in \polMatSpace$ be nonsingular and $\minDegs \in \N^\matdim$
  denote the degrees of the diagonal entries of the Hermite form $\hermite$ of
  $\mat{A}$. If $\mat{R}$ is a $-\minDegs$-column reduced form of
  $\mat{A}$, then $\mat{R}$ has $-\minDegs$-column degree
  $\cdeg[-\minDegs]{\mat{R}} = \unishift$, row degree $\rdeg{\mat{R}} =
  \minDegs$, and $\hermite = \mat{R} \, \cdot
  \leadingMat[{-\minDegs}]{\mat{R}}^{-1}$.
\end{corollary}
\begin{proof}
  Note that $\leadingMat[-\minDegs]{\hermite}$ is the identity matrix, so
  that $\hermite$ is a $-\minDegs$-reduced form of $\mat{A}$. Furthermore,
  $\hermite$ has $-\minDegs$-column degree $(0,\ldots,0)$ which implies that
  $\cdeg[-\minDegs]{\mat{R}} = \unishift$ and thus $\rdeg{\mat{R}} \le
  \minDegs$ componentwise. By \Cref{lem:popov_normalize} we
  obtain $\hermite = \mat{R} \, \cdot \leadingMat[{-\minDegs}]{\mat{R}}^{-1}$. In
  addition, we must have $\rdeg{\mat{R}} = \minDegs$, since otherwise
  $\leadingMat[-\minDegs]{\mat{R}}$ would have a zero row.
\end{proof}

Thus we can start with the matrix $\shiftMat{\shifts[\mu]-\minDegs} \mat{A}$,
column reduce this matrix and then normalize it to get our normal form.  However
$\shiftMat{\shifts[\mu]-\minDegs} \mat{A}$ may have some entries of large
degree. Indeed, $\max(\minDegs)$ may be as large as $\deg(\det(\mat{A}))$ while
having $\min(\minDegs)=0$, in which case the degree of
$\shiftMat{\shifts[\mu]-\minDegs} \mat{A}$ is at least $\deg(\det(\mat{A}))$.
For efficient deterministic shifted column reduction we would need the degree
of $\shiftMat{\shifts[\mu]-\minDegs} \mat{A}$ to be in $\bigO{\deg(\mat{A})}$.

\subsection{Reducing the amplitude of \texorpdfstring{$\minDegs$}{delta} using partial linearization}
\label{subsec:hermite_knowndeg}

In the strategy presented in the previous subsection, the main obstacle to obtaining an efficient
algorithm is that the diagonal degrees of $\hermite$ might have a large
amplitude. In this subsection, we will show how \emph{partial linearization}
techniques allow us to build a matrix $\rowParLinDiag{\mat{A}}$ such that
$\hermite$ can be obtained from a $-\shifts[d]$-reduced form of
$\rowParLinDiag{\mat{A}}$ for a shift $\shifts[d]$ that has a small amplitude.

A key fact is that the average of the degrees $\minDegs$ is controlled.
Namely, denoting by $\degExp$ the average of $\minDegs$, we have that $\degExp
\le \deg(\mat{A})$. Indeed, the product of the diagonal entries of $\hermite$
is $\det(\mat{H})$ which, up to a constant multiplier, is the same as
$\det(\mat{A})$ and thus the degree of this product is $$\matdim \degExp =
\minDeg_1 + \cdots + \minDeg_\matdim = \deg(\det(\mat{A})) \le \matdim
\deg(\mat{A}).$$ In order to reduce the amplitude of $\minDegs$, one can split
the entries that are larger than $\degExp$ into several entries each at most
$\degExp$. From this we obtain another tuple $\shifts[d] =
(\shift[d]{1},\ldots,\shift[d]{\expand{\matdim}})$  with $\max(\shifts[d]) -
\min(\shifts[d]) \le \degExp \le \deg(\mat{A})$ and having length
$\expand{\matdim}$ less than $2\matdim$.

Most importantly for our purpose, there is a corresponding transformation of
matrices which behaves well with regards to shifted reduction. Namely, this
transformation is a type of \emph{row partial linearization}~\cite[Section
6]{GSSV2012}. Let us consider the case of the Hermite form $\hermite$ of
$\mat{A}$. For each $i$, we consider the row $i$ of $\hermite$. If its degree
$\minDeg_i$ is larger than $\degExp$ then the row is expanded into $\quoExp_i$
rows of degree at most $\degExp$. This yields a $\expand{\matdim} \times
\matdim$ matrix $\widetilde{\hermite}$ of degree at most $\degExp$.
Furthermore, certain elementary columns are inserted into
$\widetilde{\hermite}$ resulting in a square nonsingular matrix
$\rowParLinDiag{\hermite}$ which preserves fundamental properties of $\hermite$
(for example, its Smith factors and its determinant). The matrix
$\rowParLinDiag{\hermite}$ has dimension $\expand{\matdim} \times
\expand{\matdim}$ and degree at most $\degExp$, which in this case is the
average row degree of $\hermite$.

Consider for example a $4\times 4$ matrix $\hermite$ in Hermite form with
diagonal entries having degrees $(2,37,7,18)$. Such a matrix has degree profile
\[\hermite = \begin{bmatrix}
  (2 ) &      &      &      \\
  [36] & (37) &      &      \\
  [ 6] & [ 6] & ( 7) &      \\
  [17] & [17] & [17] & (18)  
\end{bmatrix}, \]
where $[d]$ stands for an entry of degree at most $d$ and $(d)$ stands for a
monic entry of degree exactly $d$. Here $\hermite$ has row degree $\minDegs =
(2,37,7,18)$.

Let us now construct the row partial linearization $\rowParLinDiag{\hermite}$.
Considering the upper bound  $\degExp = 1 + \lfloor (2+37+7+18)/4 \rfloor = 17$
on the average row degree of $\hermite$, we will split the high-degree rows of
$\hermite$ in several rows having degree less than $\degExp$. The first row is
unchanged; the second row is expanded into two rows of degree $16$ and one row of
degree $3$; the third row is unchanged; and finally the last row is expanded into
one row of degree $16$ and one row of degree $1$.
The matrix with expanded rows is then
\[\expand{\hermite} = \begin{bmatrix}
  (2 ) &      &      &      \\
  [16] & [16] &      &      \\
  [16] & [16] &      &      \\
  [2 ] & (3 ) &      &      \\
  [6 ] & [6 ] & (7 ) &      \\
  [16] & [16] & [16] & [16] \\
  [0 ] & [0 ] & [0 ] & (1 )
\end{bmatrix}. \]
Note that $\hermite$ and $\expand{\hermite}$ are related by
$\expandMat \cdot \expand{\hermite} = \hermite$, where $\expandMat$ is the so-called
\emph{expansion-compression} matrix
\[\expandMat = \begin{bmatrix}
   1 & 0 &     0     &     0     & 0 & 0 &    0      \\
   0 & 1 & \var^{17} & \var^{34} & 0 & 0 &    0      \\
   0 & 0 &     0     &     0     & 1 & 0 &    0      \\
   0 & 0 &     0     &     0     & 0 & 1 & \var^{17} 
 \end{bmatrix}. \]

We can insert elementary columns in $\expand{\hermite}$ by
\[\rowParLinDiag{\hermite} = \begin{bmatrix}
  (2 )                                                  \\
  [16] & \var^{17} &        & [16]                         \\
  [16] &  -1    & \var^{17} & [16]                         \\
  [2]  &        &   -1   & (3 )                         \\
  [6 ] &        &        & [6 ] & (7 )                  \\
  [16] &        &        & [16] & [16] & \var^{17} & [16]  \\
  [0 ] &        &        & [0 ] & [0 ] &   -1   & (1 ) 
\end{bmatrix} 
\]
which indicate the row operations needed to keep track of the structure of the
original rows of $\mat{H}$.  Now the reduced tuple of row degrees $\shifts[d] =
(2,17,17,3,7,17,1)$ has as its largest entry the \emph{average} row degree
$\degExp=17$ of $\hermite$. Furthermore, $\hermite$ can be reconstructed from
$\rowParLinDiag{\hermite}$, without field operations, as a submatrix of
$\expandMat \cdot \rowParLinDiag{\hermite}$.

\begin{remark}
  \label{rk:def_parlin}
  This partial linearization differs from that of \cite[Theorem~10]{GSSV2012}
  in that
  \begin{itemize}
    \setlength\itemsep{0pt}
    \item it operates on the rows rather than the columns,
    \item it scales the inserted elementary columns by $-1$ compared to the
      elementary rows in \cite{GSSV2012}, and
    \item it keeps the linearized rows together. This reflects the fact that
      the expansion-compression matrix $\expandMat$ is a column permutation of
      the one in the construction of \cite{GSSV2012}, which would be
      \[
        \begin{bmatrix}
       1 & 0 & 0 & 0 &     0     &     0     &    0      \\
       0 & 1 & 0 & 0 & \var^{17} & \var^{34} &    0      \\
       0 & 0 & 1 & 0 &     0     &     0     &    0      \\
       0 & 0 & 0 & 1 &     0     &     0     & \var^{17} 
       \end{bmatrix} \]
       in the example above.
  \end{itemize}
  Our motivation for these changes is that the partial row linearization we use
  here preserves shifted reduced and shifted Popov forms. This will be detailed below.
  \qedexmp
\end{remark}

Formally we define the partial linearization for a matrix $\mat{A}$ and a tuple
$\minDegs$, with the latter not necessarily related to $\rdeg{\mat{A}}$.
Indeed, we will apply this in a situation where the tuple $\minDegs$ is formed
by the diagonal degrees of the Hermite form of $\mat{A}$.

\begin{definition}
  \label{defn:rowparlindiag}
  Let $\mat{A} \in \polMatSpace$, $\minDegs =
  (\minDeg_1,\ldots,\minDeg_\matdim) \in \N^\matdim$ and set
  \[\degExp = 1 + \left\lfloor \frac{(\minDeg_1 + \cdots + \minDeg_\matdim)}{\matdim}
  \right\rfloor .\] For any $i\in\{1,\ldots,\matdim\}$ write $\minDeg_i =
  (\quoExp_i -1) \degExp + \remExp_i$ with $\quoExp_i = \lceil \minDeg_i /
  \degExp \rceil$ and $1 \le \remExp_i \le \degExp$ if $\minDeg_i > 0$, while
  $\quoExp_i = 1$ and $\remExp_i = 0$ if $\minDeg_i = 0$. Set $\expand{\matdim}
  = \quoExp_1 + \cdots + \quoExp_\matdim$ and define $\shifts[d] \in
  \N^{\expand{\matdim}}$ as
  \begin{equation}
    \label{eqn:expandMinDegs}
    \shifts[d] = ( ~\underbrace{\degExp, ~\ldots~, ~\degExp,~\remExp_1}_{\quoExp_1}, ~ \ldots~,~
    \underbrace{~\degExp, ~\ldots, ~\degExp,~\remExp_\matdim}_{\quoExp_\matdim} ~)
  \end{equation}
  as well as the row expansion-compression matrix $\expandMat \in
  \polMatSpace[\matdim][\expand{\matdim}]$ as
  \begin{equation}
    \label{eqn:expandMat}
    \expandMat = 
    \begin{bmatrix}
      1 & \var^\degExp & \cdots & \var^{(\quoExp_1-1)\degExp}   \\
        &           &        &    &  \ddots              \\
        &           &        &    &          & 1 & \var^\degExp & \cdots & \var^{(\quoExp_\matdim-1)\degExp} 
    \end{bmatrix}.
  \end{equation}
  Let $\expand{\mat{A}} \in \polMatSpace[\expand{\matdim}][\matdim]$ be such
  that $\mat{A} = \expandMat \cdot \expand{\mat{A}}$ with all the rows of
  $\expand{\mat{A}}$ having degree at most $\degExp$ except possibly at
  indices $\{\quoExp_1 + \cdots + \quoExp_i, 1\le i\le \matdim\}$. 
  Define $\rowParLinDiag{\mat{A}} \in \polMatSpace[\expand{\matdim}]$ as:
  \bigskip
  
  \noindent
  \begin{itemize}
    \item[(i)] for $1\le i \le \matdim$, the column $\quoExp_1 + \cdots +
      \quoExp_i$ of $\rowParLinDiag{\mat{A}}$ is the column $i$ of $\expand{\mat{A}}$; \bigskip
    \item[(ii)] for $0\le i \le \matdim-1$ and $1 \le j \le \quoExp_{i+1}-1$, the
      column $\quoExp_1 + \cdots + \quoExp_i + j$ of $\rowParLinDiag{\mat{A}}$
      is the column $$\trsp{[0, \cdots, 0, \var^\degExp, -1, 0, \cdots, 0]} \in
      \polMatSpace[\expand{\matdim}][1]$$ with the entry $\var^\degExp$ at row
      index $\quoExp_1 + \cdots + \quoExp_i + j$.
  \end{itemize}
\end{definition}

It follows from this construction that any matrix $\mat{A} \in \polMatSpace$ is
the submatrix of $\expandMat \cdot \, \rowParLinDiag{\mat{A}}$ formed by its columns
at indices $\{\quoExp_1+\cdots+\quoExp_i, 1\le i\le \matdim\}$. 

It is important to note that this transformation has good properties regarding
the computation of $-\minDegs$-shifted reduced forms of $\mat{A}$, where
$\minDegs$ is the tuple of diagonal degrees of the Hermite form of $\mat{A}$.
Indeed, it transforms any $-\minDegs$-reduced form $\mat{R}$ of $\mat{A}$ into
a $-\shifts[d]$-reduced form $\rowParLinDiag{\mat{R}}$ of the transformed
$\rowParLinDiag{\mat{A}}$. In other words, we have the following diagram:
\[
  \begin{array}{ccccc}
    \shiftMat{\shifts[\mu]-\minDegs}\mat{A}
    & \xrightarrow{\;\;\text{reduction}\;\;}
    & -\minDegs\text{-reduced form of } \mat{A} \\
    | & & | \\
    \text{\scriptsize partial linearization} & & \text{\scriptsize partial linearization} \\
    \downarrow & & \downarrow \\
    \shiftMat{\shiftDegExp-\shifts[d]} \rowParLinDiag{\mat{A}} & 
    \xrightarrow{\;\;\text{reduction}\;\;} &
    -\shifts[d]\text{-reduced form of } \rowParLinDiag{\mat{A}}
    \end{array} \; ,
\]
where $\shiftDegExp$ is the uniform tuple
$(\max(\shifts[d]),\ldots,\max(\shifts[d]))$ of length $\expand{\matdim}$. In
terms of efficiency, it is more interesting to perform the reduction step on
$\shiftMat{\shiftDegExp-\shifts[d]} \rowParLinDiag{\mat{A}}$ with the shift
$-\shifts[d]$, rather than on $\mat{A}$ with the shift $-\minDegs$. Indeed,
using the fastest known deterministic reduction algorithm~\cite{GSSV2012}, the
latter computation uses $\softO{\matdim^\expmatmul (\deg(\mat{A}) +
\max(\minDegs))}$ field operations.  On the other hand, the former is in
$\softO{\matdim^\expmatmul (\deg(\mat{A})+\degExp)}$, since $\max(\shifts[d])
\le \degExp$ and $\deg(\rowParLinDiag{\mat{A}}) \le \deg(\mat{A})$. We recall
that $\degExp$ is close to the average of $\minDegs$.

We state this formally in the following lemma.  For the sake of presentation we postpone the proof until later in
\Cref{subsec:properties_proof}.

\begin{lemma}
  \label{lem:rowparlindiag}
  Let $\minDegs = (\minDeg_1,\ldots,\minDeg_\matdim) \in \N^\matdim$, and
  define $\shifts[d]$ as in \cref{eqn:expandMinDegs}.
  \begin{enumerate}[(i)]
    \item If a matrix $\mat{R} \in \polMatSpace$ is $-\minDegs$-reduced with
      $-\minDegs$-column degree $\unishift$, then $\rowParLinDiag{\mat{R}}$ is
      $-\shifts[d]$-reduced with $-\shifts[d]$-column degree $\unishift$.
    \item If two matrices $\mat{A}$ and $\mat{B}$ in $\polMatSpace$ are right
      unimodularly equivalent, then $\rowParLinDiag{\mat{A}}$ and
      $\rowParLinDiag{\mat{B}}$ are also right unimodularly equivalent.
    \item If $\mat{A} \in \polMatSpace$ is nonsingular, $\mat{R}$ is a
      $-\minDegs$-reduced form of $\mat{A}$, and $\mat{R}$ has
      $-\minDegs$-column degree $\unishift$, then $\rowParLinDiag{\mat{R}}$ is
      a $-\shifts[d]$-reduced form of $\rowParLinDiag{\mat{A}}$ with
      $-\shifts[d]$-column degree $\unishift$.
  \end{enumerate}
\end{lemma}

Our algorithm will first build $\rowParLinDiag{\mat{A}}$ and then find a
$-\shifts[d]$-reduced form $\hat{\mat{R}}$ for this new matrix. We note that,
for any $-\minDegs$-reduced form $\mat{R}$ of $\mat{A}$, the matrix
$\hat{\mat{R}} = \rowParLinDiag{\mat{R}}$ is a suitable reduced form and, as
remarked earlier, has the property that it is easy to recover $\mat{R}$.
However, it is not the case that any $\hat{\mat{R}}$ computed by shifted
reduction from $\rowParLinDiag{\mat{A}}$ will have the form $\hat{\mat{R}} =
\rowParLinDiag{\mat{R}}$. In order to solve this issue, we will rely on
normalization as in \Cref{lem:popov_normalize}.  This allows us to deduce
$\rowParLinDiag{\hermite}$ from $\hat{\mat{R}}$, and then the entries of
$\hermite$ can be read off from those of $\rowParLinDiag{\hermite}$.
Diagrammatically we have
\[
  \begin{array}{ccccc}
    \shiftMat{\shifts[\mu]-\minDegs}\mat{A}
    & \xrightarrow{\:\;\text{reduction}\;\:}
    & \shiftMat{\shifts[\mu]-\minDegs}\mat{R}
    & \xrightarrow{\;\:\text{normalization}\;\:}
    & \hermite = \mat{R} \, \cdot \leadingMat[-\minDegs]{\mat{R}}^{-1} \\
    | & & & & | \\
    \text{\scriptsize partial linearization} & & & & \text{\scriptsize partial linearization} \\
    \downarrow & & & & \downarrow \\
    \shiftMat{\shiftDegExp-\shifts[d]} \rowParLinDiag{\mat{A}} & 
    \xrightarrow{\:\;\text{reduction}\;\:} &
    \shiftMat{\shiftDegExp-\shifts[d]} \hat{\mat{R}} &
    \xrightarrow{\;\:\text{normalization}\;\:} & 
    \rowParLinDiag{\hermite} 
    = \hat{\mat{R}} \,\cdot \leadingMat[{-\shifts[d]}]{\hat{\mat{R}}}^{-1}
    \end{array} \; .
\]

\begin{corollary}
  \label{cor:reduction_pivdeg}
  Let $\mat{A} \in \polMatSpace$ be nonsingular and let $\minDegs =
  (\minDeg_1,\ldots,\minDeg_\matdim) \in \N^\matdim$ denote the degrees of the
  diagonal entries of the Hermite form $\hermite$ of $\mat{A}$. Using the notation
  from \Cref{defn:rowparlindiag}, we have that
  \begin{enumerate}[(i)]
    \item $\leadingMat[{-\shifts[d]}]{\rowParLinDiag{\hermite}}$ is the
      identity matrix,
    \item if $\hat{\mat{R}} \in \polMatSpace[\expand{\matdim}]$ is a
      $-\shifts[d]$-reduced form of $\rowParLinDiag{\mat{A}}$, then
      $\rowParLinDiag{\hermite} = \hat{\mat{R}} \, \cdot
      \leadingMat[{-\shifts[d]}]{\hat{\mat{R}}}^{-1}$.
  \end{enumerate}
\end{corollary}
\begin{proof}
  $(i)$ follows from the construction of $\rowParLinDiag{\hermite}$. From
  \Cref{lem:rowparlindiag} we have that $\rowParLinDiag{\hermite}$ is a
  $-\shifts[d]$-reduced form of $\mat{A}$, so that $(ii)$ follows from $(i)$
  and \Cref{lem:popov_normalize}.
\end{proof}

In particular, $\hermite$ can be recovered as being the submatrix of
$\expandMat \, \cdot  \hat{\mat{R}} \, \leadingMat[{-\shifts[d]}]{\hat{\mat{R}}}^{-1}$
formed by its columns $\{\quoExp_1+\cdots+\quoExp_i, 1\le i\le \matdim\}$.

\begin{example}[Reducing the diagonal degrees]
  \label{eg:reducing_mindeg}
  Consider a matrix $\mat{A} \in \polMatSpace[4]$ such that its Hermite form
  $\hermite$ has diagonal degrees $\minDegs = (2,37,7,18)$.
  As shown earlier,
  \[\rowParLinDiag{\hermite} = \begin{bmatrix}
    (2 )                                                  \\
    [16] & \var^{17} &        & [16]                         \\
    [16] &  -1~    & \var^{17} & [16]                         \\
    [2]  &        &   -1~   & (3 )                         \\
    [6 ] &        &        & [6 ] & (7 )                  \\
    [16] &        &        & [16] & [16] & \var^{17} & [16]  \\
    [0 ] &        &        & [0 ] & [0 ] &   -1~   & (1 ) 
  \end{bmatrix}. \]
  We see that $\shifts[d]=(2,17,17,3,7,17,1)$ corresponds to the row degree of
  $\rowParLinDiag{\hermite}$, that this matrix has $-\shifts[d]$-column
  degree $\unishift$ and that its $-\shifts[d]$-leading matrix is the identity.
  In particular, it is $-\shifts[d]$-reduced. In addition, from $(ii)$ of
  \Cref{lem:rowparlindiag}, $\rowParLinDiag{\hermite}$ and
  $\rowParLinDiag{\mat{A}}$ are right-unimodularly equivalent. As a
  result, $\rowParLinDiag{\hermite}$ is a $-\shifts[d]$-reduced form of
  $\rowParLinDiag{\mat{A}}$.

  Let $\hat{\mat{R}}$ be any $-\shifts[d]$-reduced form of
  $\rowParLinDiag{\mat{A}}$. Then $\hat{\mat{R}}$ also has $-\shifts[d]$-column
  degree $\unishift$, its $-\shifts[d]$-leading matrix is invertible, and its
  degree profile is 
  \[\hat{\mat{R}} = \begin{bmatrix}
    [2 ] & [2 ] & [2 ] & [2 ] & [2 ] & [2 ] & [2 ] \\
    [17] & [17] & [17] & [17] & [17] & [17] & [17] \\
    [17] & [17] & [17] & [17] & [17] & [17] & [17] \\
    [3 ] & [3 ] & [3 ] & [3 ] & [3 ] & [3 ] & [3 ] \\
    [7 ] & [7 ] & [7 ] & [7 ] & [7 ] & [7 ] & [7 ] \\
    [17] & [17] & [17] & [17] & [17] & [17] & [17] \\
    [1 ] & [1 ] & [1 ] & [1 ] & [1 ] & [1 ] & [1 ]
  \end{bmatrix}. \]
  While $\hat{\mat{R}}$ is generally not of the form $\rowParLinDiag{\mat{R}}$
  for $\mat{R}$ some $-\minDegs$-reduced form of $\mat{A}$, it still follows
  from \Cref{lem:popov_normalize} that $\rowParLinDiag{\hermite} =
  \hat{\mat{R}} \, \cdot \leadingMat[{-\shifts[d]}]{\hat{\mat{R}}}^{-1}$. \qedexmp
\end{example}

\subsection{Algorithm and computational cost}
\label{subsec:algo_proof}

The results in the previous subsection lead to
\Cref{algo:hermite_knowndeg} for the computation of the Hermite form $\hermite$ from $\mat{A}$ and
$\minDegs$. Its main computational task is to compute a column reduced form of
a matrix of dimension $\bigO{\matdim}$ and degree $\bigO{\deg(\mat{A})}$
(Step~12). This can be done efficiently and deterministically with the
algorithm in~\cite[Section~8]{GSSV2012}.

\begin{proposition}
  \label{prop:knowndeg_hermite}
  Let $\mat{A} \in \polMatSpace$ be nonsingular, and let $\minDegs \in
  \N^\matdim$ be the degrees of the diagonal entries of the Hermite form of
  $\mat{A}$. On input $\mat{A}$ and $\minDegs$,
  \Cref{algo:hermite_knowndeg} computes the Hermite form of $\mat{A}$
  using $\softO{\matdim^\expmatmul \deg(\mat{A})}$ field operations.
\end{proposition}
\begin{proof}
  The correctness of the algorithm follows directly from
  \Cref{cor:reduction_pivdeg} and from the remark that a matrix $\mat{R} \in
  \polMatSpace[\expand{\matdim}]$ is $-\shifts[d]$-column reduced if and only
  if $\mat{D} \cdot \mat{R}$ is column reduced (for the uniform shift), where
  $\mat{D}$ is the diagonal matrix at Step~11.

  Furthermore, we have $\deg(\mat{D}) \le \degExp$ and
  $\deg(\rowParLinDiag{\mat{A}}) \le \max(\deg(\mat{A}),\degExp)$. Since 
  $\degExp = 1 + \lfloor \sumVec{\minDegs} / \matdim \rfloor$, and
  as $\hermite$ is in Hermite form and $\minDegs$ are the degrees of its
  diagonal entries, we have $\sumVec{\minDegs} = \deg(\det(\hermite)) =
  \deg(\det(\mat{A})) \le \matdim \deg(\mat{A})$. Thus, $\degExp \le 1 +
  \deg(\mat{A})$ and the degrees of $\mat{D}$ and
  $\rowParLinDiag{\mat{A}}$ are both at most $1 + \deg(\mat{A})$. Their product $\mat{D} \cdot \rowParLinDiag{\mat{A}}$ therefore has degree at most $2 +
  2\deg(\mat{A})$. On the other hand, these matrices have dimension
  \[\expand{\matdim} = \sum_{i=1}^\matdim \quoExp_i \le \sum_{i=1}^\matdim (1 +
    \minDeg_i / \degExp) = \matdim + \frac{\sumVec{\minDegs}}{1 + \lfloor
    \sumVec{\minDegs}/\matdim \rfloor} < 2\matdim.\] As a result, Step~12 uses
    $\softO{\matdim^\expmatmul \deg(\mat{A})}$ field
    operations~\cite[Theorem~18]{GSSV2012}.

  Concerning Step~13, from \Cref{cor:reduction_pivdeg} the matrix
  $\hat{\mat{R}}$ has row degree $\shifts[d]$, and
  $\leadingMat[{-\shifts[d]}]{\hat{\mat{R}}}^{-1}$ is a constant matrix. Thus
  the computation of $\hat{\mat{R}} \cdot
  \leadingMat[{-\shifts[d]}]{\hat{\mat{R}}}^{-1}$ can be performed via complete
  linearization of the rows of $\hat{\mat{R}}$, using $\bigO{\matdim^\expmatmul
    \lceil \sumVec{\shifts[d]} / \matdim \rceil}$ operations. This concludes
    the proof since $\sumVec{\shifts[d]} = \sumVec{\minDegs} =
    \deg(\det(\mat{\hermite})) = \deg(\det(\mat{A})) \le \matdim
    \deg(\mat{A})$.
\end{proof}

\begin{algorithm}[t]
  \caption{$\mathrm{HermiteKnownDegree}(\mat{A},\shifts,\minDegs)$}
\label{algo:hermite_knowndeg}

\begin{algorithmic}[1]
\REQUIRE{$\mat{A}\in\polMatSpace$ a nonsingular matrix, $\minDegs =
(\minDeg_1,\ldots,\minDeg_\matdim) \in \N^\matdim$ the degrees of the diagonal
entries of the Hermite form of $\mat{A}$.}

\ENSURE{the Hermite form of $\mat{A}$.}

\STATE{$\degExp := 1 + \lfloor
(\minDeg_1 + \cdots + \minDeg_\matdim) / \matdim \rfloor$;}

\FOR{$i\in \{1,\ldots,\matdim\}$}
\IF{$\minDeg_i>0$}
\STATE{$\quoExp_i := \lceil \degExp / \minDeg_i \rceil$; $\remExp_i := \minDeg_i - (\quoExp_i-1) \degExp$;}
\ELSE
\STATE{$\quoExp_i := 1$; $\remExp_i = 0$;}
\ENDIF
\ENDFOR

\STATE{$\expand{\matdim} := \quoExp_1 + \cdots + \quoExp_\matdim$ and
$\expandMat \in \matSpace[\expand{\matdim}][\matdim]$ as in
\cref{eqn:expandMat}; }

\STATE{$\shifts[d] = (\shift[d]{1},\ldots,\shift[d]{\expand{\matdim}})$
as in \cref{eqn:expandMinDegs};}

\STATE{$\mat{D} := \mathrm{Diag}(\var^{\degExp-\shift[d]{1}},\ldots,\var^{\degExp-\shift[d]{\expand{\matdim}}})$;}

\STATE{$\mat{D} \hat{\mat{R}} :=$ column reduced form of $\mat{D} \cdot \rowParLinDiag{\mat{A}}$;}
\hfill \COMMENT{\emph{using the algorithm in~\cite{GSSV2012}}}

\STATE{$\hat{\hermite} := \expandMat \, \cdot \hat{\mat{R}}
\, \cdot \leadingMat[{-\shifts[d]}]{\hat{\mat{R}}}^{-1}$;}

\STATE{$\hermite :=$ the submatrix of $\hat{\hermite}$ formed by its columns
$\{\quoExp_1+\cdots+\quoExp_i, 1\le i\le \matdim\}$}

\RETURN{$\hermite$;}
\end{algorithmic}
\end{algorithm}

Combining \Cref{algo:hermite_knowndeg,alg:hermiteDiagonal} results in a deterministic 
algorithm for computing the Hermite form of $\mat{A}$
in $\softO{\matdim^\expmatmul \deg(\mat{A})}$ field operations. 

\begin{example}
  \label{eg:ctd_knowndeg}
  Let $\field = \Z_7$ be the field with $7$ elements, and consider the matrix
  $\mat{A} \in \polMatSpace[3]$ from \Cref{eg:ctd_hermitediagonal}:
  \[
  \mat{A} = 
  \begin{bmatrix}
                    6x + 1 &       2x^3 + x^2 + 6x + 1 &                      3 \\
    4x^5 + 5x^4 + 4x^2 + x &    6x^5 + 5x^4 + 2x^3 + 4 & x^4 + 5x^3 + 6x^2 + 5x \\
                         2 & 2x^5 + 5x^4 + 5x^3 + 6x^2 &                      6
  \end{bmatrix} .
  \]
  According to \Cref{eg:ctd_hermitediagonal} the diagonal entries of the
  Hermite form of $\mat{A}$ have degrees $\minDegs = (0,1,9)$. Note that
  $\minDegs$ is non-uniform, and $\max(\minDegs) - \min(\minDegs) =
  \deg(\det(\mat{A})) - 1$.
  
  Using the column reduction algorithm in \cite{GSSV2012} to compute a
  $-\minDegs$-reduced form of $\mat{A}$ would imply working on the matrix
  $\shiftMat{\shifts[\mu]-\minDegs}\mat{A} = \shiftMat{(9,8,0)}\mat{A}$, which
  has degree $13 = \deg(\det(\mat{A})) + \deg(\mat{A}) - 2$. In this case
  partial linearization gives us a $5 \times 5$ matrix
  $\rowParLinDiag{\mat{A}}$ and a shift $\shifts[d]$ such that
  $\deg(\rowParLinDiag{\mat{A}}) \le \deg(\mat{A})$ and $\max(\shifts[d]) -
  \min(\shifts[d]) \le \deg(\mat{A})$. In particular, the matrix
  $\shiftMat{\shifts[m]-\shifts[d]}\rowParLinDiag{\mat{A}}$ to be reduced has
  degree $8 \le 2 \deg(\mat{A})$.
  
  To see this, \Cref{defn:rowparlindiag} gives the parameters $\degExp=4$,
  $\tuple{\quoExp} = (1,1,3)$, $\tuple{\remExp} = (0,1,1)$, $\shifts[d] = (0,
  1, 4, 4, 1)$, the expansion-compression matrix
  \[
    \expandMat =
    \begin{bmatrix}
      1 & 0 & 0 &   0 &   0 \\
      0 & 1 & 0 &   0 &   0 \\
      0 & 0 & 1 & x^4 & x^8
    \end{bmatrix},
  \]
  and finally
  \[
    \rowParLinDiag{\mat{A}} =
    {
    \begin{bmatrix}
                   6x + 1 &    2x^3 + x^2 + 6x + 1 &   0 &   0 &                      3 \\
   4x^5 + 5x^4 + 4x^2 + x & 6x^5 + 5x^4 + 2x^3 + 4 &   0 &   0 & x^4 + 5x^3 + 6x^2 + 5x \\
                        2 &            5x^3 + 6x^2 & x^4 &   0 &                      6 \\
                        0 &                 2x + 5 &   6 & x^4 &                      0 \\
                        0 &                      0 &   0 &   6 &                      0
    \end{bmatrix}.
    }
  \]
  Computing a $-\shifts[d]$-reduced form for $\rowParLinDiag{\mat{A}}$ gives
  \[
    \hat{\mat{R}} =
    {
    \begin{bmatrix}
                    5 &                    1 &   0 &                   1 &                          2 \\
                    5 &               4x + 4 &   0 &              3x + 5 &                     6x + 3 \\
       x^3 + 6x^2 + 4 &    3x^4 + x^3 + 6x^2 & x^4 & x^3 + 5x^2 + 4x + 3 & 6x^4 + 2x^3 + 3x^2 + x + 6 \\
      3x^3 + 4x^2 + 6 & 4x^4 + 4x^3 + 4x + 5 &   6 &        x^3 + 2x + 4 &       5x^4 + 2x^3 + 4x + 2 \\
                    6 &                    x &   0 &                   6 &                          0
    \end{bmatrix}
    } .
  \]
  Note that $\rdeg{\mat{R}} = \shifts[d]$, and
  more precisely,
  \[
    \leadingMat[{-\shifts[d]}]{\mat{R}}
    =
    \begin{bmatrix}
      5 & 1 & 0 & 1 & 2 \\
      0 & 4 & 0 & 3 & 6 \\
      0 & 3 & 1 & 0 & 6 \\
      0 & 4 & 0 & 0 & 5 \\
      0 & 1 & 0 & 0 & 0
    \end{bmatrix}.
  \]
  Normalizing $\hat{\mat{R}}$ via $\hat{\mat{R}} \cdot
    \leadingMat[{-\shifts[d]}]{\mat{R}}^{-1}$ gives
  \begin{align*}
    \rowParLinDiag{\hermite} 
    & =
    \begin{bmatrix}
                    1 &                    0 &   0 &   0 &                   0 \\
                    1 &                x + 6 &   0 &   0 &                   0 \\
      3x^3 + 4x^2 + 5 & 4x^3 + 5x^2 + 6x + 4 & x^4 &   0 &    3x^3 + 3x^2 + 4x \\
      2x^3 + 5x^2 + 4 &     2x^3 + 3x^2 + 3x &   6 & x^4 & x^3 + 4x^2 + 6x + 4 \\
                    4 &                    3 &   0 &   6 &               x + 2
    \end{bmatrix}.
  \end{align*}
  Performing the inverse linearization, by taking columns $(1,2,5)$ of
  $\expandMat \cdot \rowParLinDiag{\hermite}$, directly gives the entries in the
  Hermite form of $\mat{A}$:
  \[
    \hermite =
    \begin{bmatrix}
            1 &      0 &                                                        0 \\
            1 &  x + 6 &                                                        0 \\
       h_{31} & h_{32} & x^9 + 2x^8 + x^7 + 4x^6 + 6x^5 + 4x^4 + 3x^3 + 3x^2 + 4x
    \end{bmatrix}
  \]
  with
  \begin{align*}
    h_{31} & = 4x^8 + 2x^7 + 5x^6 + 4x^4 + 3x^3 + 4x^2 + 5  ,    \\
    h_{32} & = 3x^8 + 2x^7 + 3x^6 + 3x^5 + 4x^3 + 5x^2 + 6x + 4 .
    \qedexmp
  \end{align*}
  
\end{example}

\subsection{Proof of \texorpdfstring{\Cref{lem:rowparlindiag}}{Lemma 5.5}}
\label{subsec:properties_proof}

Let us now give the detailed proof of \Cref{lem:rowparlindiag}. 

$(i)$
Since $\mat{R} \in \polMatSpace$ is $-\minDegs$-reduced with
$-\minDegs$-column degree $\unishift$, it has row degree $\minDegs$ since otherwise
the invertible matrix $\leadingMat[-\minDegs]{\mat{R}}$ would have a zero row.
We show that $\leadingMat[{-\shifts[d]}]{\rowParLinDiag{\mat{R}}}$ is a
permutation of the rows and columns of $\begin{bmatrix}
  \leadingMat[-\minDegs]{\mat{R}} & \mat{0} \\ \mat{0} & \idMat \end{bmatrix}
\in \matSpace[\expand{\matdim}]$. In particular,
$\leadingMat[{-\shifts[d]}]{\rowParLinDiag{\mat{R}}}$ is invertible and thus
$\rowParLinDiag{\mat{R}}$ is $-\shifts[d]$-reduced.

Let us first observe it on an example. We consider the case $\minDegs =
(2,37,7,18)$. Then $\mat{R}$ has the following degree profile,
\[\mat{R} = \begin{bmatrix}
    [2 ] & [2 ] & [2 ] & [2 ] \\
    [37] & [37] & [37] & [37] \\
    [ 7] & [ 7] & [ 7] & [ 7] \\
    [18] & [18] & [18] & [18]  
\end{bmatrix} \]
with invertible $-\minDegs$-leading matrix. Following the construction in
\Cref{defn:rowparlindiag}, we have $\shifts[d] = (2,17,17,3,7,17,1)$
and
\[\rowParLinDiag{\mat{R}} = \begin{bmatrix}
  [2 ] &        &        & [2 ] & [2 ] &        & [2 ] \\
  [16] & \var^{17} &        & [16] & [16] &        & [16] \\
  [16] &  -1    & \var^{17} & [16] & [16] &        & [16] \\
  [3 ] &        &   -1   & [3 ] & [3 ] &        & [3 ] \\
  [7 ] &        &        & [7 ] & [7 ] &        & [7 ] \\
  [16] &        &        & [16] & [16] & \var^{17} & [16] \\
  [1 ] &        &        & [1 ] & [1 ] &   -1   & [1 ] 
\end{bmatrix}. \]
Observe that $\mat{R}$ has $-\shifts[d]$-column degree at most $\unishift$
componentwise, and that its $-\shifts[d]$-leading matrix is
\[\leadingMat[{-\shifts[d]}]{\rowParLinDiag{\mat{R}}} = \begin{bmatrix}
  \ell_{11} &   &   & \ell_{12} & \ell_{13} &   & \ell_{14} \\
            & 1 &   &                                       \\
            &   & 1 &                                       \\
  \ell_{21} &   &   & \ell_{22} & \ell_{23} &   & \ell_{24} \\
  \ell_{31} &   &   & \ell_{32} & \ell_{33} &   & \ell_{34} \\
            &   &   &           &           & 1 &           \\
  \ell_{41} &   &   & \ell_{42} & \ell_{42} &   & \ell_{42}
\end{bmatrix}, \]
where $(\ell_{ij})_{1\le i,j \le 4} = \leadingMat[-\minDegs]{\mat{R}}$. Since
$\leadingMat[-\minDegs]{\mat{R}}$ is invertible,
$\leadingMat[{-\shifts[d]}]{\rowParLinDiag{\mat{R}}}$ is invertible as well. Furthermore  $\rowParLinDiag{\mat{R}}$ is $-\shifts[d]$-reduced and that
it has $-\shifts[d]$-column degree $\unishift$.

In the general case, by construction of $\rowParLinDiag{\mat{R}}$ one can
check that $\leadingMat[{-\shifts[d]}]{\rowParLinDiag{\mat{R}}}$ is a matrix in
$\matSpace[\expand{\matdim}]$ such that 
\begin{itemize}
  \item[(a)] its $\matdim \times \matdim$ submatrix with row and column indices in
$\{\quoExp_1+\cdots+\quoExp_i, 1\le i\le \matdim\}$ is equal to
$\leadingMat[-\minDegs]{\mat{R}}$,
  \item[(b)] its $(\expand{\matdim}-\matdim) \times (\expand{\matdim}-\matdim)$
    submatrix with row and column indices in $\{1,\ldots,\expand{\matdim}\} -
    \{\quoExp_1+\cdots+\quoExp_i, 1\le i\le \matdim\}$ is equal to the
    identity matrix,
  \item[(c)] its other entries are all zero.
\end{itemize}
This directly implies that
$\leadingMat[{-\shifts[d]}]{\rowParLinDiag{\mat{R}}}$ is invertible. In addition by construction
$\rowParLinDiag{\mat{R}}$ has $-\shifts[d]$-column degree at most
$\unishift$ componentwise.  The fact that
$\leadingMat[{-\shifts[d]}]{\rowParLinDiag{\mat{R}}}$ is invertible also implies
that $\rowParLinDiag{\mat{R}}$ has $-\shifts[d]$-column degree exactly
$\unishift$.

\smallskip
$(ii)$ Denote by $\triExpMat \in
\polMatSpace[\expand{\matdim}][\expand{\matdim}-\matdim]$ the submatrix of
$\rowParLinDiag{\mat{A}}$ formed by its columns at indices $\{\quoExp_1 +
  \cdots + \quoExp_i + j, 1 \le j \le \quoExp_{i+1}-1, 0\le i \le \matdim-1
\}$. Up to a permutation of its columns, $\rowParLinDiag{\mat{A}}$ is then
$[\triExpMat \;\; \widetilde{\mat{A}}]$. In particular, $\expandMat \, \cdot
\rowParLinDiag{\mat{A}}$ is right-unimodularly equivalent to $\expandMat \,
[\triExpMat \;\; \widetilde{\mat{A}}] = [\mat{0} \;\; \mat{A}]$. For the remainder of this proof we 
will use the shorthand notation $\expandMat \, \cdot
\rowParLinDiag{\mat{A}} \equiv [\mat{0} \;\; \mat{A}]$.

Define the matrix $\mat{E} \in
\matSpace[(\expand{\matdim}-\matdim)][\expand{\matdim}]$ whose row $\quoExp_1 +
\cdots + \quoExp_i + j - i$ is the coordinate vector with $1$ at index
$\quoExp_1 + \cdots + \quoExp_i + j + 1$, for all $1 \le j \le
\quoExp_{i+1}-1$ and $0\le i \le \matdim-1$. That is, we have
\[
  \begin{bmatrix} \mat{E} \\ \expandMat \end{bmatrix} = 
  \begin{bmatrix}
    0 &       1                                               \\
      &              & \ddots &                               \\
      &              &        &            1                  \\
      &              &        &                             & \ddots \\
      &              &        &                             &        & 0 &     1                                               \\
      &              &        &                             &        &   &            & \ddots &                               \\
      &              &        &                             &        &   &            &        &            1                  \\
    1 & \var^\degExp & \cdots & \var^{(\quoExp_1-1)\degExp}   \\
      &              &        &                             & \ddots               \\
      &              &        &                             &        & 1 & \var^\degExp & \cdots & \var^{(\quoExp_\matdim-1)\degExp} 
  \end{bmatrix}.
\]
By construction, the matrix $\mat{U} := \mat{E} \cdot \triExpMat$ is upper
triangular with diagonal entries $-1$, and thus unimodular. As a result,
\[
  \begin{bmatrix}
    \mat{E} \\
    \expandMat
  \end{bmatrix}
  \rowParLinDiag{\mat{A}}
  \;\;\equiv\;\;
  \begin{bmatrix}
    \mat{E} \\
    \expandMat
  \end{bmatrix}
  \begin{bmatrix}
    \triExpMat & \expand{\mat{A}}
  \end{bmatrix}
  =
  \begin{bmatrix}
    \mat{U} & \anyMat \\
    \mat{0} & \mat{A}
  \end{bmatrix}
  \;\;\equiv\;\;
  \begin{bmatrix} 
    \idMat & \mat{0} \\ \mat{0} & \mat{A}
  \end{bmatrix}.
\]
Similarly, we have that 
$\begin{bmatrix}
    \mat{E} \\
    \expandMat
  \end{bmatrix}
\rowParLinDiag{\mat{B}} \equiv
\begin{bmatrix} \idMat & \mat{0} \\ \mat{0} & \mat{B} \end{bmatrix}$.

Since $\mat{A} \equiv \mat{B}$ by assumption, we obtain
$\begin{bmatrix}
    \mat{E} \\
    \expandMat
  \end{bmatrix}
\rowParLinDiag{\mat{A}}
\equiv
\begin{bmatrix}
    \mat{E} \\
    \expandMat
  \end{bmatrix}
\rowParLinDiag{\mat{B}}$.
This implies that $\rowParLinDiag{\mat{A}} \equiv \rowParLinDiag{\mat{B}}$
since the matrix $\begin{bmatrix} \mat{E} \\ \expandMat \end{bmatrix}$ is
invertible (more precisely, its determinant is $1$).

\medskip
$(iii)$ is a direct consequence of $(i)$ and $(ii)$.

\section{Reduction to almost uniform degrees in Hermite form computation}
\label{sec:parlin}

As mentioned in \Cref{subsec:degdet}, we aim at a cost bound which involves the
generic determinant bound. In \Cref{sec:diagonals} we showed how to compute
the diagonal entries of $\hermite$ in $\softO{\matdim^\expmatmul \lceil
s\rceil}$ operations, with $s$ the average column degree of the input matrix.
However, this does not take into account the fact that the degrees of its rows
are possibly unbalanced. In \Cref{sec:hermite}, we were only able to obtain the
cost bound $\softO{\matdim^\expmatmul \deg(\mat{A})}$ for computing the
remaining entries of $\hermite$.

The goal of this section is to show that, applying results
from~\cite[Section~6]{GSSV2012}, one can give a reduction from the general case
of Hermite form computation to the case where the degree of the input matrix
$\mat{A}$ is in $\bigO{\lceil \degDet/\matdim \rceil}$. This is stated formally
in \cref{prop:reduction_hermite}, after what we give two complete examples to
illustrate this reduction (\Cref{exmp:reducing_degrees,eg:ctd_parlin}).

To get a rough idea of how the partial linearization in
\cite[Section~6]{GSSV2012} works and how it benefits Hermite form computation,
consider the matrix
\[   \mat{A} = \begin{bmatrix} 1 & x^{39} + x \\ x & x^{41} + 1 \end{bmatrix}.\]
In this case the column degrees of the matrix are quite unbalanced as $1$ and
$41$ have an average column degree of $21$. However we can create
a second  matrix, of slightly larger dimension, as
\[
  \mat{B} = \left[ \begin{array}{ccc} 1 & 0 & -x^{22} \\ 
x^{17} & 1 & x \\ x^{19} & x & 1 
\end{array} \right]
\]
which shares some nice properties with $\mat{A}$. This matrix is constructed by
dividing the third column into its two $x^{22}$-adic coefficients (rows $2$ and
$3$) and then including an additional row (row $1$) which provides the single
column operation which would undo the division. Thus by construction this
matrix is unimodularly equivalent to 
\[
 \left[ \begin{array}{ccc} 1 & 0 & 0 \\
x^{17} & 1 & x^{39} + x \\ x^{19} & x & x^{41} + 1 
\end{array} \right]
\]
and it is easily seen that the Hermite form of $\mat{A}$ will be given by the
$2 \times 2$ trailing submatrix of the Hermite form of $\mat{B}$. As such we
rely on the computation of the Hermite form of a matrix, not much larger than
the original matrix, but having the nice property that the degrees are much
more uniformly distributed.

\begin{proposition}
  \label{prop:reduction_hermite}
  Let $\mat{A} \in \polMatSpace$ be nonsingular. Using no operation in
  $\field$, one can build a nonsingular matrix $\mat{B} \in \polMatSpace[m]$
  such that
  \begin{enumerate}[(i)]
    \item $\matdim \le m < 3 \matdim$ and $\deg(\mat{B}) \le \lceil\degDet /
      \matdim\rceil$,
    \item the Hermite form of $\mat{A}$ is the trailing principal
      $\matdim \times \matdim$ submatrix of the Hermite form of $\mat{B}$.
  \end{enumerate}
\end{proposition}

\begin{proof}
  The partial linearization used in  \cite[Corollary 3]{GSSV2012} takes $\mat{A}$
  and constructs a matrix $\mat{C} \in \polMatSpace[m]$ with smoothed degrees 
  having the properties: (a) $\mat{C}$ is a nonsingular matrix with $m < 3\matdim$, 
  (b) $\deg(\mat{C}) \le \lceil \degDet / \matdim \rceil$ and (c) the principal $\matdim \times \matdim$ submatrix of $\mat{C}^{-1}$ is equal to
  $\mat{A}^{-1}$.
  Permuting the rows and columns of this matrix $\mat{C}$ into
  \[
    \mat{B} \;=\;
    \begin{bmatrix}
      \mat{0} & \idMat_{m-\matdim} \\
      \idMat_{\matdim} & \mat{0}
    \end{bmatrix}
    \mat{C} 
    \begin{bmatrix}
      \mat{0} & \idMat_{\matdim} \\
      \idMat_{m-\matdim} & \mat{0}
    \end{bmatrix}
    \in \polMatSpace[m],
  \]
  we see that  $\mat{A}^{-1}$ appears as the trailing $\matdim \times \matdim$
  submatrix of $\mat{B}^{-1}$. We will prove that the Hermite form of
  $\mat{B}$ has the shape
  $\begin{bmatrix}
    \idMat  & \mat{0}  \\
    \anyMat & \hermite
  \end{bmatrix}$,
  where $\hermite$ is the Hermite form of $\mat{A}$.

  Let
  $\mat{T} =   
    \begin{bmatrix}
      \hermite_1 & \mat{0} \\
      \anyMat & \hermite_2 
    \end{bmatrix}$
  be the Hermite form of $\mat{B}$, where $\hermite_1 \in
  \polMatSpace[(m-\matdim)]$ and $\hermite_2 \in \polMatSpace[\matdim]$. We
  can write $\hermite_2 = \mat{A} \mat{D}$, where the matrix $\mat{D} =
  \mat{A}^{-1} \hermite_2$ has entries in $\polRing$. Indeed, by
  construction,
  \[
    \mat{B}^{-1} \mat{T} = 
    \begin{bmatrix}
      \anyMat  &   \anyMat     \\
      \anyMat  &  \mat{A}^{-1}
    \end{bmatrix}
    \begin{bmatrix}
      \hermite_1 & \mat{0} \\
      \anyMat & \hermite_2 
    \end{bmatrix}
    = 
    \begin{bmatrix}
      \anyMat  &   \anyMat     \\
      \anyMat  &  \mat{A}^{-1} \hermite_2
    \end{bmatrix}.
  \] 
  is a (unimodular) matrix in $\polMatSpace[m]$.  On the other hand, according
  to \cite[Corollary~5]{GSSV2012} we have $\det(\mat{B}) = \det(\mat{A})$, and
  therefore 
  \[
    \det(\mat{A}) = \lambda \det(\mat{T}) = \lambda \det(\hermite_1) \det(\hermite_2) = \lambda \det(\hermite_1) \det(\mat{A}) \det(\mat{D})
  \]
  where $\lambda = \det(\mat{B}^{-1}\mat{T})^{-1}$ is a nonzero constant from
  $\field$. Thus, $\hermite_1$ and $\mat{D}$ are both unimodular. Therefore, since $\hermite_1$ is in Hermite form, it must be
  the identity matrix and, since $\hermite_2$ is in Hermite form and
  right-unimodularly equivalent to $\mat{A}$, it must be equal to $\hermite$.
\end{proof}

For the details of how to build the matrix $\mat{C}$ using row and column
partial linearization, we refer the reader to \cite[Section~6]{GSSV2012}. We
give here two detailed examples (see also \cite[Example~4]{GSSV2012}), written
with the help of our prototype implementation of the algorithms described in
this paper.

\begin{example}
  \label{exmp:reducing_degrees}
  Let $\field$ be the finite field with $997$ elements. Using a computer
  algebra system, we choose $\mat{A}\in \polMatSpace[4]$ with prescribed
  degrees and random coefficients from $\field$.  Instead of showing the entire
  matrix let us only consider the degree profile which in this case is
  \[ \mat{A} =
    \begin{bmatrix}
      [ 2] & [10] & [63] & [ 5] \\
      [75] & [51] & [95] & [69] \\
      [ 4] & [ 5] & [48] & [ 7] \\
      [10] & [54] & [75] & [ 6]
  \end{bmatrix}, \]
  where $[d]$ indicates an entry of degree $d$. For the sake of presentation,
  we note that $\degDet = 199 = 75 + 54 + 63 + 7$; however, this quantity is
  not computed by our algorithm. Instead, to find which degrees we will use to
  partially linearize the columns of $\mat{A}$, we permute its rows and columns
  to ensure that the diagonal degrees dominate the degrees in the trailing
  principal submatrices:
  \begin{equation}
    \label{eqn:degree_profile}
    \begin{bmatrix}
      0 & 1 & 0 & 0 \\
      0 & 0 & 0 & 1 \\
      0 & 0 & 1 & 0 \\
      1 & 0 & 0 & 0
    \end{bmatrix}
    \begin{bmatrix}
      [ 2] & [10] & [63] & [ 5] \\
      [75] & [51] & [95] & [69] \\
      [ 4] & [ 5] & [48] & [ 7] \\
      [10] & [54] & [75] & [ 6]
    \end{bmatrix}
    \underbrace{
    \begin{bmatrix}
      0 & 0 & 0 & 1 \\
      0 & 1 & 0 & 0 \\
      1 & 0 & 0 & 0 \\
      0 & 0 & 1 & 0
    \end{bmatrix}
    }_{\boldsymbol{\pi}}
    \;=\;
    \begin{bmatrix}
      [95] & [51] & [69] & [75]  \\
      [75] & [54] & [ 6] & [10]  \\ 
      [48] & [ 5] & [ 7] & [ 4]  \\
      [63] & [10] & [ 5] & [ 2]
    \end{bmatrix}
  \end{equation}
  Then, the diagonal degrees $95,54,7,2$ are used for column partial
  linearization; we remark that $95 + 54 + 7 + 2 = 158 \le \degDet$. Permuting
  back the rows and columns of $\mat{A}$, we will partially linearize its
  columns with respect to the degrees $(2,54,95,7) = (95,54,7,2)\boldsymbol{\pi}^{-1}$.  Since the average of
  these degrees is $\lceil 158/4 \rceil = 40$, the columns are linearized
  into $(1,2,3,1)$ columns, respectively. That is, columns $1$ and $4$ of
  $\mat{A}$ will not be affected, column $2$ of $\mat{A}$ will be expanded into
  $2$ columns, and column $3$ of $\mat{A}$ will be expanded into $3$ columns.
  Elementary rows are inserted at the same time to reflect these column
  linearizations.  Thus, we obtain a column linearized version of $\mat{A}$ as
  \[ \hat{\mat{A}} =
    \begin{bmatrix}
      [ 2]  & [10]       & [39]       & [ 5]  &   0  & [23]       &   0  \\
      [75]  & [39]       & [39]       & [69]  & [11] & [39]       & [15] \\
      [ 4]  & [ 5]       & [39]       & [ 7]  &   0  & [ 8]       &   0  \\
      [10]  & [39]       & [39]       & [ 6]  & [14] & [35]       &   0  \\
      0     & -\var^{40} & 0          & 0     & 1    & 0          & 0    \\
      0     & 0          & -\var^{40} & 0     & 0    & 1          & 0    \\
      0     & 0          & 0          & 0     & 0    & -\var^{40} & 1    
  \end{bmatrix}. \]

  In particular, we have
  \[ 
    \rdeg{\hat{\mat{A}}} = (39, 75, 39, 39, 40, 40, 40),
  \]
  whose average is $\lceil 312 / 7\rceil = 45$. Now, we perform a partial
  linearization on the rows with respect to their row degree. Only the second
  row has degree $75>45$, and is therefore split into two rows; inserting an
  elementary column accordingly, we obtain
  \[ \mat{C} =
    \begin{bmatrix}
      [ 2] & [10]       & [39]       & [ 5]  &   0  & [23]       &   0  &  0         \\
      [44] & [39]       & [39]       & [44]  & [11] & [39]       & [15] &  -\var^{45}\\
      [ 4] & [ 5]       & [39]       & [ 7]  &   0  & [ 8]       &   0  &  0         \\
      [10] & [39]       & [39]       & [ 6]  & [14] & [35]       &   0  &  0         \\
        0  & -\var^{40} &   0        &   0   &   1  &   0        &   0  &  0         \\
        0  &   0        & -\var^{40} &   0   &   0  &   1        &   0  &  0         \\
        0  &   0        &   0        &   0   &   0  & -\var^{40} &   1  &  0         \\
      [30] &   0        &   0        & [24]  &   0  &   0        &   0  &  1         
  \end{bmatrix} \]
  whose degree is $45$. Finally, we verify that the Hermite form of
  $\begin{bmatrix}
    \mat{0} & \idMat_4 \\
    \idMat_4 & \mat{0}
  \end{bmatrix}
  \mat{C}
  \begin{bmatrix}
    \mat{0} & \idMat_4 \\
    \idMat_4 & \mat{0}
  \end{bmatrix}$
  is
  $\begin{bmatrix}
    \idMat_4 & \mat{0} \\
    \anyMat & \hermite
  \end{bmatrix}$, 
  with $\hermite$ the Hermite form of $\mat{A}$. Thus, we have transformed
  a Hermite form problem in dimensions $4 \times 4$ and degree $95$ into one in
  dimensions $8 \times 8$ but degree less than $50 = \lceil\degDet / 4\rceil$.
  \qedexmp
\end{example}

\begin{example}
  \label{eg:ctd_parlin} 
Let $\field$ be the field with $7$ elements, and consider the matrix from
\Cref{eg:ctd_hermitediagonal}:
\[
  \mat{A} = 
  \begin{bmatrix}
                    6x + 1 &       2x^3 + x^2 + 6x + 1 &                      3 \\
    4x^5 + 5x^4 + 4x^2 + x &    6x^5 + 5x^4 + 2x^3 + 4 & x^4 + 5x^3 + 6x^2 + 5x \\
                         2 & 2x^5 + 5x^4 + 5x^3 + 6x^2 &                      6
  \end{bmatrix} .
\]
Here, $\degDet = \deg(\det(\mat{A})) = 1+5+4 = 10$. We consider a row- and
column-permuted version of the matrix $\mat{A}$ ensuring that the diagonal
degrees are dominant, as we did in \cref{exmp:reducing_degrees}:
\[
  \begin{bmatrix}
    0 & 1 & 0 \\
    1 & 0 & 0 \\
    0 & 0 & 1
  \end{bmatrix}
  \mat{A}
  \underbrace{
  \begin{bmatrix}
    0 & 1 & 0 \\
    1 & 0 & 0 \\
    0 & 0 & 1
  \end{bmatrix}
  }_{\boldsymbol{\pi}}
  = 
  \begin{bmatrix}
       6x^5 + 5x^4 + 2x^3 + 4 & 4x^5 + 5x^4 + 4x^2 + x & x^4 + 5x^3 + 6x^2 + 5x \\
          2x^3 + x^2 + 6x + 1 &                 6x + 1 &                      3 \\
    2x^5 + 5x^4 + 5x^3 + 6x^2 &                      2 &                      6
  \end{bmatrix} .
\]
This gives us the linearization degrees $(1,5,0) =
(5,1,0)\boldsymbol{\pi}^{-1}$, which have average $\lceil 6/3\rceil = 2$,
so the partial column linearization results in
\[
  \hat{\mat{A}} =
  \begin{bmatrix}
                      6x + 1 & 6x + 1 &                      3 & 2x + 1 &      0 \\
      4x^5 + 5x^4 + 4x^2 + x &      4 & x^4 + 5x^3 + 6x^2 + 5x &     2x & 6x + 5 \\
                           2 &      0 &                      6 & 5x + 6 & 2x + 5 \\
                           0 &   6x^2 &                      0 &      1 &      0 \\
                           0 &      0 &                      0 &   6x^2 &      1
  \end{bmatrix}.
\]
Then, we perform row partial linearization of this matrix with respect to its row
degrees $(1,5,1,2,2)$, whose average is $\lceil 11 / 5\rceil = 3$, giving
\[
  \mat{C} =
  \begin{bmatrix}
                      6x + 1 & 6x + 1 &                      3 & 2x + 1 &      0 & 0 \\
                    4x^2 + x &      4 &              6x^2 + 5x &     2x & 6x + 5 & 6x^3 \\
                           2 &      0 &                      6 & 5x + 6 & 2x + 5 & 0 \\
                           0 &   6x^2 &                      0 &      1 &      0 & 0 \\
                           0 &      0 &                      0 &   6x^2 &      1 & 0 \\
                   4x^2 + 5x &      0 &                    x+5 &      0 &      0 & 1
  \end{bmatrix}.
\]

Using the algorithm in \Cref{sec:diagonals}, we obtain the degrees
$(0,0,0,0,1,9)$ of the diagonal entries of the Hermite form of the permuted
matrix
\[
  \mat{B} =
  \begin{bmatrix}
    \mat{0} & \idMat_{3} \\
    \idMat_{3} & \mat{0}
  \end{bmatrix}
  \mat{C}
  \begin{bmatrix}
    \mat{0} & \idMat_{3} \\
    \idMat_{3} & \mat{0}
  \end{bmatrix}.
\]
Proceeding then as in \Cref{sec:hermite}, we can to compute the complete
Hermite form of $\mat{B}$ using the knowledge of these degrees, giving
\[
  \begin{bmatrix}
    \idMat_3 & \mat{0} \\
    \mat{R} & \hermite
  \end{bmatrix}
\]
where $\hermite$ is the Hermite form of $\mat{A}$ as given in \Cref{eg:ctd_knowndeg}, and
the transpose of $\mat{R}$ is
\[\trsp{\mat{R}} =
  {
  \begin{bmatrix}
    0 & 6 & 4x^7 + 6x^6 + x^5 + 4x^4 + 2x^3 + 6x^2 + 4 \\
    0 & 3 & 6x^8 + 4x^7 + 4x^5 + 3x^4 + 3x^3 + 2x + 6  \\
    0 & 4 & 3x^8 + 2x^7 + 3x^6 + 3x^5 + 4x^3 + 5x^2 + x + 2
  \end{bmatrix}
  }.
\qedexmp
\]
\end{example}

\section{Conclusion\label{sec:Future-Research}}

In this paper we have given new, deterministic algorithms for computing the
Hermite normal form and the determinant of a nonsingular polynomial matrix. Our
methods are based on the efficient, deterministic computation of the diagonal
elements of the Hermite form. While our algorithms are fast in terms of the
number of operations in an abstract field ${\mathbb K}$, they do not take into
consideration the possible growth of coefficients in the field, an issue when
working over certain fields such as ${\mathbb Q}$, the rational numbers.
Kannan \cite{kannan} was the first to show that computing Hermite normal forms
over ${\mathbb Q}[x]$ can be done in polynomial time. Fraction-free algorithms
for Hermite form computation which take into consideration coefficient growth
have been given in \cite{BLV:jsc06} (using a shifted Popov algorithm) and
\cite{labhalla} (where the problem is converted into a large linear system). We
plan to investigate exact algorithms for Hermite and determinant computation
based on the fraction-free approach used in \cite{BL2000} and also the use of 
Chinese remaindering. In the latter case the reduced domains (e.g.
$\Z_p[x]$) do not encounter coefficient growth which allows for effective use
of the algorithms in this paper. The issue in this case is the reconstruction
of the images, where we expect the techniques used in \cite{cheng} will be
helpful. 

In terms of additional future research we are interested in the still open
problem of reducing computation of the Hermite form over the integers
\cite{storjohann-labahn96} to the complexity of integer matrix multiplication.
In addition, we are interested in finding efficient, deterministic algorithms
for other normal forms of polynomial matrices, such as the Popov normal form,
or more generally the shifted Popov normal forms.  In addition we are
interested in fast normal form algorithms where the entries are differential
operators rather than polynomials. Such algorithms are useful for reducing
systems of linear differential equations to solvable first order systems
\cite{barkatou:2013}. 

\bigskip \noindent\textbf{Acknowledgments.} The authors would like to thank
Arne Storjohann and an anonymous referee for suggestions on simplifying the
presentation of \cref{sec:parlin} and about the alternative determinant
algorithm in the Appendix. G.~Labahn was supported by a grant from NSERC 
while V. Neiger was supported by the international mobility
grants from \emph{Projet Avenir Lyon Saint-\'Etienne}, \emph{Mitacs Globalink -
Inria}, and Explo'ra Doc from \emph{R\'egion Rh\^one-Alpes}.

\appendix

\bibliographystyle{siamplain}

\section*{Appendix. Another fast and deterministic algorithm for the determinant}
\label{app:alternative_det}

In this appendix, we describe an alternative to our determinant
\cref{alg:determinant}, kindly suggested by a reviewer. The main idea is to
rely on $x$-Smith decomposition \cite{GSSV2012} in order to make sure that the
determinant can be easily retrieved from the diagonal entries of a triangular
form computed with \cref{alg:hermiteDiagonal}. This is thus a way to overcome
the obstacle mentioned in \cref{rk:det_from_diagonal}.

Let $\mat{A} \in \polMatSpace$ be a nonsingular polynomial matrix. Then,
\cite[Corollary~1]{GSSV2012} states that we can compute a triangular $x$-Smith
decomposition of $\mat{A}$ using $\softO{\matdim^\expmatmul \deg(\mat{A})}$
field operations. This yields matrices $\pi, \mat{U},\hermite$ such that
$\mat{A} \pi = \mat{U} \hermite$, where
\begin{itemize}
  \setlength\partopsep{0pt}
  \setlength\itemsep{1pt}
  \item $\pi \in \matSpace[\matdim]$ is a permutation matrix,
  \item $\hermite \in \polMatSpace[\matdim]$ is triangular with
    $\det(\hermite) = \var^{\alpha}$ for some $\alpha\in\N$,
  \item $\mat{U} \in \polMatSpace[\matdim]$ is such that $\det(\mat{U} \bmod
    \var) \neq 0$ and $\deg(\mat{U}) \le \deg(\mat{A})$.
\end{itemize}
Then, we have $\det(\mat{A}) = \det(\mat{U}) \det(\hermite) \det(\pi)^{-1}$,
and the cost of finding $\det(\hermite)$ and $\det(\pi)$ is negligible. It
remains to compute $\det(\mat{U})$, which can be done in
$\softO{\matdim^\expmatmul \deg(\mat{U})} \subseteq \softO{\matdim^\expmatmul
  \deg(\mat{A})}$ operations. Indeed, since $\det(\mat{U} \bmod \var) \neq 0$,
  determining the diagonal entries of a triangular form of $\mat{U}$ allows us
  to deduce its determinant as explained in \cref{rk:det_from_diagonal}.

Thus, we obtain $\det(\mat{A})$ in $\softO{\matdim^\expmatmul \deg(\mat{A})}$
field operations; with \cref{prop:reduction_determinant}, this gives another
proof of \cref{thm:determinant}.
\end{document}